\documentclass[onecolumn, draftclsnofoot,11pt]{IEEEtran}
\IEEEoverridecommandlockouts
\usepackage{mathrsfs}
\usepackage{tikz}
\usepackage[utf8]{inputenc}
\usepackage{amsfonts} 
\usepackage{epstopdf}
\usepackage{graphicx}
\usepackage{bbm}
\usepackage{enumerate}
\usepackage{epsf,epsfig,verbatim,amssymb,amsmath,array,cite,amsthm,subfigure,multicol,multirow}  
\usepackage{psfrag,xspace}

\usepackage{bm}
\usepackage{hhline}
\usepackage{mathabx}
\usepackage{physics}
\usepackage{xcolor}
\definecolor{darkblue}{rgb}{0.1,0.1,0.8}
\definecolor{brickred}{rgb}{0.8, 0.25, 0.33}
\definecolor{DarkGreen}{rgb}{0,0.6,0}

\newtheorem{theorem}{Theorem}
\newtheorem{lem}{Lemma}

\newtheorem{proposition}{Proposition}

\theoremstyle{definition}
\newtheorem{definition}{Definition}
\newtheorem{example}{Example}

\newtheorem*{prob*}{Problem}

\allowdisplaybreaks 
\makeatletter
\def\old@comma{,}
\catcode`\,=13
\def,{%
  \ifmmode%
    \old@comma\discretionary{}{}{}%
  \else%
    \old@comma%
  \fi%
}
\makeatother

\usepackage{etoolbox}
\providetoggle{Blue_revision}
\settoggle{Blue_revision}{true}

 \usepackage{hyperref}
\hypersetup{
colorlinks=true, %
 pdfstartview={FitH},
    linkcolor=red,
    citecolor=blue, 
    urlcolor={blue!80!black}
}

\usepackage{graphicx}

\def\ulineX{\underline{X}}
\def\ulinex{\underline{x}}

\def\ulineW{\underline{W}}
\def\ulinew{\underline{w}}

\def\ulineCalW{\underline{\mathcal{W}}}

\def\ulineX{\underline{X}}

\def\CalX{\mathcal{X}}
\def\CalY{\mathcal{Y}}

\def\FF{\mathbb{F}}
\def\EE{\mathbb{E}}
\def\PP{\mathbb{P}}
\def\11{\mathbbm{1}}

\def\UCodeword{\mathtt{w}_1(a_1,m_1,\mu_1)}
\def\VCodeword{\mathtt{w}_2(a_2,m_2,\mu_2)}

\def\TDelta{\mathcal{T}_{\delta}^{(n)}}

\usepackage{stackengine}
\def\deq{\mathrel{\ensurestackMath{\stackon[1pt]{=}{\scriptstyle\Delta}}}}
\def\define{\mathrel{\ensurestackMath{\stackon[1pt]{=}{\scriptstyle\Delta}}}}

\newcommand*{\medcup}{\mathbin{\scalebox{1.2}{\ensuremath{\cup}}}}%

\def\gammaCoeff{\gamma_{w_1^n}^{(\mu_1)}} \def\zetaCoeff{\zeta_{w_2^n}^{(\mu_2)}} 

\newcounter{relctr} 
\everydisplay\expandafter{\the\everydisplay\setcounter{relctr}{0}} 

\newcommand\labelrel[2]{%
  \begingroup
    \refstepcounter{relctr}%
    \stackrel{\textnormal{(\alph{relctr})}}{\mathstrut{#1}}%
    \originallabel{#2}%
  \endgroup
}
\AtBeginDocument{\let\originallabel\label} 

\begin{document}

\title{\huge Synthesizing Correlated Randomness using Algebraic Structured Codes}


\author{\IEEEauthorblockN{Touheed Anwar Atif\IEEEauthorrefmark{1},
Arun Padakandla\IEEEauthorrefmark{2} and  S. Sandeep Pradhan\IEEEauthorrefmark{1} \\}
\IEEEauthorblockA{Department of Electrical Engineering and Computer Science,\\
\IEEEauthorrefmark{1}University of Michigan, Ann Arbor, MI 48109, USA.\\
\IEEEauthorrefmark{2}University of Tennessee, Knoxville, USA\\
Email: \tt touheed@umich.edu, arunpr@utk.edu, pradhanv@umich.edu}}




\maketitle

\begin{abstract}

In this problem, Alice and Bob, are provided $X_{1}^{n}$ and $X_{2}^{n}$ that are IID $p_{X_1 X_2}$. Alice and Bob can communicate to Charles over (noiseless) links of rate $R_1$ and $R_2$, respectively. Their goal is to enable Charles generate samples $Y^{n}$ such that the triple $(X_{1}^{n},X_{2}^{n},Y^{n})$ has a PMF that is close, in total variation, to $\prod p_{X_1 X_2 Y}$. In addition, the three parties may posses shared common randomness at rate $C$. We address the problem of characterizing the set of rate triples $(R_1,R_2,C)$ for which the above goal can be accomplished. We build on our recent findings and propose a new coding scheme based on coset codes. We analyze its information-theoretic performance and derive a new inner bound. We identify examples for which the derived inner bound is analytically proven to contain rate triples that are not achievable via any known unstructured code based coding techniques. Our findings build on a variant of soft-covering which generalizes its  applicability to the algebraic structured code ensembles. This adds to the advancement of the use structured codes in network information theory.  
\end{abstract}

\section{Introduction} 
\label{sec:intro}

The task of generating correlated randomness at different terminals in a network finds its applications in several communication and computing paradigms. This task is also fundamental to several cryptographic protocols. In this article, we provide a new information-theoretic coding framework for generating such correlated randomness in network scenarios.

We consider the scenario which was originally studied by authors in \cite{atif2020ISITsource}, as depicted in Fig~\ref{Fig:NetworkSoftCovering}. Three distributed parties, say
Alice, Bob and Charles, have to generate samples that are independent 
and identically distributed (IID) with a target probability mass function (PMF) 
$p_{X_{1}X_{2}Y}$. Alice and Bob are provided with samples that are IID according to 
$p_{X_{1}X_{2}}$ - the marginal of the target PMF 
$p_{X_{1}X_{2}Y}$. They have access to unlimited private randomness and share 
noiseless communication links of rates $R_{1},R_{2}$ with Charles. In addition, the three parties share common randomness at rate $C$. The authors in \cite{atif2020ISITsource} provided a set of sufficient conditions, i.e., an achievable rate region for such a scenario. However, can this rate-region be improved? This article answers the above question in the affirmative. 

It is well established that traditional coding techniques using unstructured codes do not achieve optimality for the several multi-terminal scenarios. For instance, the work by K\"orner-Marton \cite{korner1979encode} demonstrated this sub-optimality for a classical distributed lossless compression problem with symmetric binary sources using random linear codes. We harness analogous gains for the problem of generating correlated randomness at distributed parties. Specifically, we propose a coding scheme based on coset codes, analyze its information-theoretic performance and thereby derive a new inner bound (see Thm.~\ref{Thm:Distributed}). We identify an example for which the derived inner bound is analytically proven to contain rate triples that are not achievable in the earlier known results \cite{atif2020ISITsource}. While the derived inner bound does not subsume the one characterized in \cite{}, one can adopt the technique in \cite[Sec.~VII]{ahlswede1983source} - also demonstrated in a related context \cite{padakandla2016achievable} - to derive an inner bound that subsumes the inner bounds derived in \cite{atif2020ISITsource} and Thm.~\ref{Thm:Distributed}.


The problem of generating correlated randomness can be traced back to Wyner \cite{197401TIT_Wyn}, whose work discovered the important technical tool, called the \textit{soft covering}. This tool has found its application in diverse fields including  cryptography and quantum information theory. The work in \cite{atif2020ISITsource} further refined this tool by introducing a joint-typicality based application. As we illustrate in the sequel, this work adds another dimension to our current understanding of soft covering, what we term as the \textit{change of measure soft covering}.


A renewed interest in soft covering led Cuff \cite{CuffPhDThesis}, \cite{201311TIT_Cuf} to consider a point-to-point (PTP) version of the scenario depicted in Fig.~\ref{Fig:NetworkSoftCovering}, wherein Bob (or $X_{2}$) is absent. A side-information based scenario was subsequently studied in \cite{yassaee2015channel} and a converse provided in \cite{atif2020ISITsource}. In \cite{atif2020ISITsource} we studied the above scenario using unstructured coding techniques. 
A similar sequence of problems were also studied in the quantum setting \cite{winter2004extrinsic,wilde2012information,201907ISIT_HeiAtiPra}.


\begin{figure}
 \centering
\includegraphics[width=2.6in]{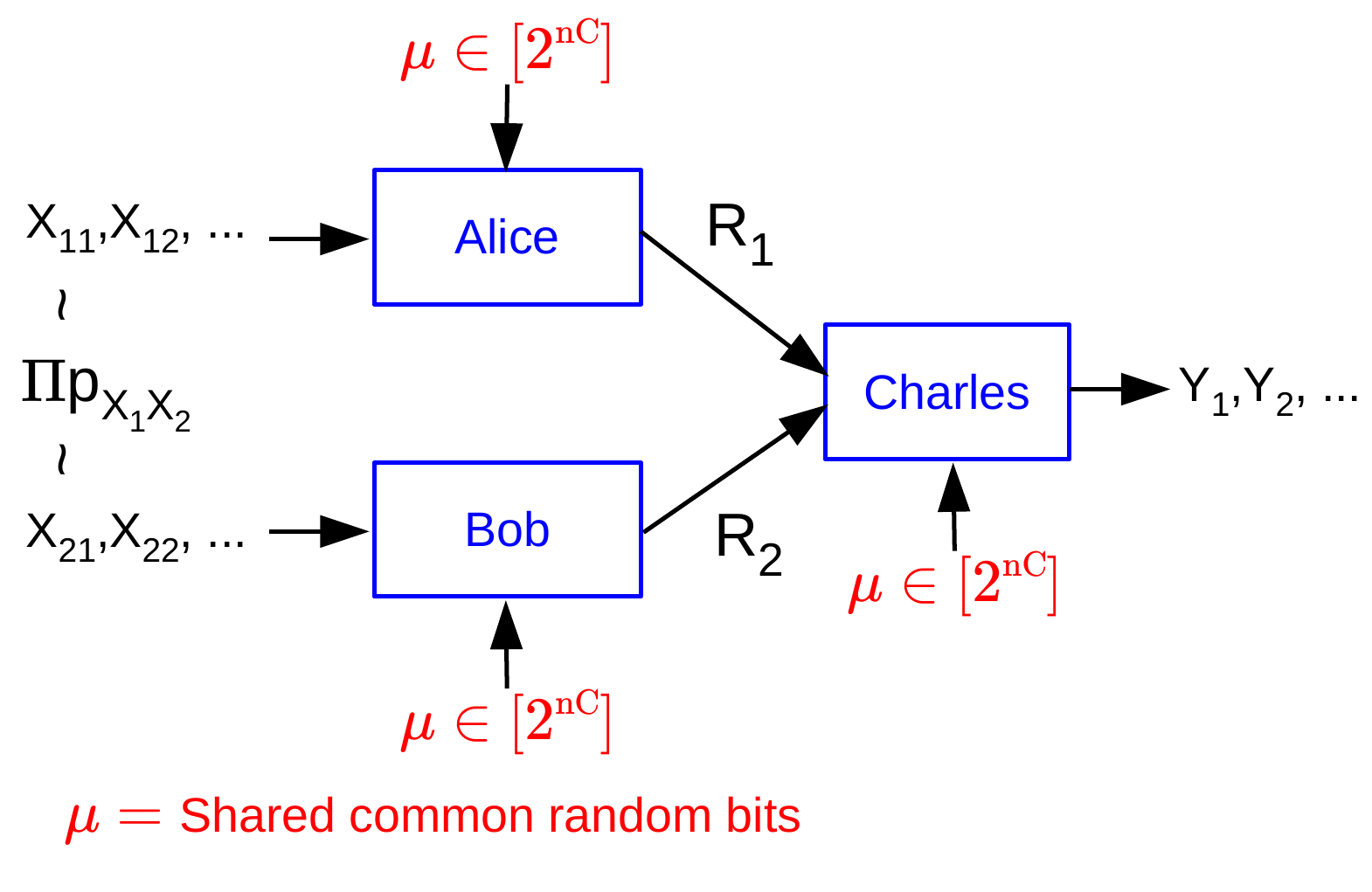}
\caption{ Source Coding for Synthesizing Correlated Randomness}
\label{Fig:NetworkSoftCovering}
\end{figure}


While all of the above works leverage the unstructured IID random codes, it has been proven that algebraic structured codes provide gains  in network communication involving distributed encoders \cite{krithivasan2011distributed,200710TIT_NazGas,200906TIT_PhiZam,201109TITarXiv_JafVis,pradhanalgebraic,padakandla2016achievable,padakandla2013computing}.
Motivated by this, we consider the distributed correlation synthesis problem depicted in Fig.~\ref{Fig:NetworkSoftCovering} and present a new achievable rate-region using structured coding techniques. 
We highlight two main challenges in this endeavour. The first challenge is to be able to achieve rates corresponding to non-uniform distributions. In particular, codewords within a random linear code has uniform empirical distributions. This requires us to enlarge our codes to be able to identify codeword with the desired single-letter distribution. We address this challenge by using a random shifts of cosets of a linear code as our code, henceforth referred to as Unionized Coset Codes (UCCs) \cite{pradhanalgebraic}.
The second challenge concerns the statistical dependence among codewords of a coset code. In contrast to IID codes, 
the codewords of a UCC are only pairwise independent \cite{Gal-ITRC68}. This prevents us from using the Chernoff concentration bound. We therefore develop novel techniques for our information theoretic study.
\section{Preliminaries and Problem Statement}
\label{Sec:Prelims}
We supplement standard information theory notation with the following. For a 
PMF $p_{X}$, we let $p^{n}_{X}=\prod_{i=1}^{n}p_{X}$. For an integer $n\geq 1$, 
$[n]\define \{1,\cdots,n\}$. The total variation between PMFs $p_X$ and $q_X$ 
defined over $\mathcal{X}$ is denoted $\|p_X-q_X\|_1 = 
\frac{1}{2}\sum_{x \in \mathcal{X}} |p_X(x) - q_X(x)|$.
 $\FF_p$ is used to denote a finite field of size $p$ with addition $\oplus.$ 

 Building on this, we address the network scenario 
(Fig.~\ref{Fig:NetworkSoftCovering}) for which we state the problem below. In the following, we let $\ulineX = (X_1,X_2),  \ulinex^n = (x^n_1,x^n_2)$.
\begin{definition}
\label{Defn:Distributed}
Given a PMF $p_{X_{1}X_{2}Y}$ on $\mathcal{X}_1\cross\mathcal{X}_2\cross\mathcal{Y}$, a rate triple $(R_1,R_2,C)$ is achievable, 
if $\forall\epsilon\!>\!0$ and sufficiently large $n$, 
there exists $2^{nC}$ randomized encoder pairs $E_{j}^{(\mu)} : 
\mathcal{X}_{j}^{n} 
\rightarrow [\Theta_{j}] : j \in [2], \mu \in [2^{nC}]$, and a 
corresponding collection of $2^{nC}$ randomized decoders 
$D^{(\mu)}:[\Theta_{1}]\times [\Theta_{2}] \rightarrow 
\mathcal{Y}^{n}$ for $\mu \in [2^{nC}]$ such that $\left| p^{n}_{\ulineX Y}- p_{\ulineX^{n} Y^{n}}\right|_{1} \leq \epsilon$, $\frac{1}{n}\log_2 \Theta_{j} \leq R_{j}+\epsilon : j \in [2]$, where
\begin{eqnarray}
p_{\ulineX^{n}Y^{n}}(\ulinex^{n},y^{n}) =\!\!\!\! \sum_{\mu \in 
[2^{nC}]}\!\!\!2^{-nC} \!\!\!\!\!\sum_{\substack{ (m_{1},m_{2}) 
\in \\ [ \Theta_{1}]\times [\Theta_{2}] }}\!\!\!\!
p^{n}_{\ulineX Y}(\ulinex^{n},y^{n})\nonumber\\p^{(\mu)}_{M_{1}|X_{1}^{n}}(m_{1}
|x_ { 1 } ^ { n } )p^ { (\mu)}_{ M_ { 2 } |X_{2}^{n}}(m_{2}|x_{2}^{n})p^{(\mu)}_ 
{ Y^{ n } |M_{1},M_{2}}(y^{n}|m_{1},m_{2}) 
\nonumber
\end{eqnarray}
$p^{(\mu)}_{M_{j}|X_{j}^{n}} : j \in [2], p^{(\mu)}_{Y^{n}|M_{1},M_{2}}$ are the 
PMFs induced by the two randomized encoders and decoder respectively, 
corresponding to common randomness message $\mu$. We 
let $\mathcal{R}_{d}(p_{\ulineX Y})$ denote the set of achievable rate triples.
\end{definition}
Theorem~\ref{Thm:Distributed} provides a new characterization of $\mathcal{R}_{d}(p_{\ulineX Y})$ based on coset codes, for the above described problem statement. This characterization provides a new inner bound to the achievable rate-region. An essential aspect of our work is the identification of a PMF $p_{X_1X_2Y}$ for which the coding scheme described in \cite{atif2020ISITsource,atif2020source}
is strictly sub-optimal.




\section{Distributed Soft Covering using Algebraic Structured Random Codes}
\label{Sec:DistribuetdSoftCovering}

\subsection{Change of Measure Soft Covering}\label{lem:changeMeasureCovering}
Before presenting the main result of the paper, we develop the necessary tools and  provide a lemma which is crucial for the upcoming results. This lemma extends the cloud mixing result of \cite{201311TIT_Cuf} with a mismatched codebook generation process.
The lemma is as follows.
\begin{lem}\label{lem:changeMeasureCoveringLemma}
Consider a PMF $p_{XY}$ on $\mathcal{X}\cross\mathcal{Y},$ and let $R$ be a finite non-negative integer. Additionally, assume that there exists some set $\Bar{\CalX}$ containing the set $\CalX$, with $p_{XY}(x,y) = 0$ for all $x \in \Bar{\CalX}\backslash\CalX$. Suppose $q_X$ is any PMF on the set $\Bar{\CalX}$ such that the PMF $p_X$ is absolutely continuous with respect to the $q_X$. Let a random code $\mathbb{C} \deq \{X^n(m): m\in[2^{nR}]\}$ be defined as a collection of codewords chosen pairwise independently from the set $\Bar{\CalX}$ according to the PMF $q_X^n$. Then we have for $R \geq H_{q}(X) - H_p(Y|X) = I_{p}(X;Y) - H_p(X) + H_q(X)$,
\begin{align}
    \lim_{n\rightarrow \infty} \EE_{\mathbb{C}}&\left[\sum_{y^n \in \CalY^n}\Big|p^n_Y(y^n) - \frac{1}{M}\sum_{m=1}^{2^{nR}}\frac{p^n_X(X^n(m))}{q^n_X(X^n(m))}p^n_{Y|X}(y^n|X^n(m))\Bigg|\right] =0\nonumber
\end{align}
\end{lem}
\begin{proof}
The proof follows similar analysis as the proof of \cite[Lemma 19]{cuff2009communication} as hence is omitted.

\end{proof}
\subsection{Main Result}
Our main result is the characterization of $\mathcal{R}_{s}(p_{\ulineX Y})$ which is the inner bound to $\mathcal{R}_{d}(p_{\ulineX Y})$. In the following, we let $\ulineX = (X_1,X_2), \ulineW = (W_1,W_2), \ulinex = (x_1,x_2)$ and $\ulinew = (w_1,w_2)$.
\begin{theorem}
\label{Thm:Distributed}
Given a PMF $p_{X_1X_2Y}$, let $\mathcal{P}(p_{X_1X_2Y})$ denote the collection of all PMFs $p_{QW_1W_2\ulineX Y}$ defined on $\mathcal{Q}\times\mathcal{W}_1\times\mathcal{W}_2\times\mathcal{\ulineX}\times\mathcal{Y}$ such that 
(i) $p_{\ulineX Y}(\ulinex,y) =  \sum_{(q,\ulinew) \in \mathcal{Q} \times\ulineCalW}p_{Q\ulineW \ulineX Y}(q,\ulinew,\ulinex,y)$
for all $(\ulinex,y) \in \underline{\CalX}\cross\CalY$, (ii) $W_1-QX_1-QX_2-W_2$ and $\ulineX-Q\ulineW-Y$ are Markov chains, (iii) $|\mathcal{W}_1|\leq |\mathcal{X}_1|$, $|\mathcal{W}_2|\leq |\mathcal{X}_2|$. Further, let $\beta(p_{Q\ulineW\ulineX Y})$ denote the set of rates and common randomness triple $(R_1,R_2,C)$ that satisfy
\begin{align}\label{eq:rate-region}
    {R}_1 & \geq I(X_1;W_1|W_2,Q) + I(W_1\oplus W_2;W_2|Q) \nonumber\\
    {R}_2 & \geq I(X_2;W_2|W_1,Q) + I(W_1\oplus W_2;W_1|Q) \nonumber \\
    {R}_1 + C & \geq I(\ulineX;W_1|W_2,Q)  + I(Y;W_1|\ulineX,Q) + I(W_1\oplus W_2;W_2|Q)  \nonumber \\
    {R}_2 + C & \geq I(\ulineX;W_2|W_1,Q)  + I(Y;W_2|\ulineX,Q) + I(W_1\oplus W_2;W_1|Q)  \nonumber \\
    R_1+R_2+C & \geq I(\ulineX;W_1|W_2,Q) + I(\ulineX;W_2|W_1,Q) + I(W_1\oplus W_2;W_1|Q) + I(W_1\oplus W_2;W_2|Q),
\end{align}
where the above information theoretic terms are evaluated with respect to the PMF $p_{QW_1W_2\ulineX Y}$. Let
\begin{eqnarray}  
\mathcal{R}_{s}(p_{\underline{X}Y}) \deq \text{Closure}
    \left(\bigcup_{p_{Q\ulineW\ulineX Y} \in \mathcal{P}(p_{X_1X_2Y})}\beta(p_{Q\ulineW\ulineX Y})\right)
\end{eqnarray}    
We have 
\[
\mathcal{R}_{s}(p_{\underline{X}Y})
\subseteq \mathcal{R}_d(p_{\ulineX Y}).
\]
In other words, the rate triple $(R_1,R_2,C) \in \left(\bigcup_{p_{Q\ulineW\ulineX Y} \in \mathcal{P}(p_{X_1X_2Y})}\beta(p_{Q\ulineW\ulineX Y})\right)$ is achievable.
\end{theorem}
Note that the rate-region obtained in Theorem 2 of  \cite{atif2020source} contains the constraint $R_1 + R_2 + C \geq I(X_1X_2Y;W_1W_2|Q)$. Hence when $2H(W_1\oplus W_2|Q) < H(W_1,W_2|Q),$ the above theorem gives a lower sum rate constraint. As a result, the rate-region above contains points that are not contained within the rate-region provided in \cite{atif2020source}. To illustrate this fact further, consider the following example.
\begin{example}
Let $X_1$ and $X_2$ be a pair of binary symmetric correlated sources with 
$P(X_2=1|X_1=0)=p$, for some $p \in (0,0.5)$. Let 
$Y=X_1 \oplus X_2 \oplus Q$, where $P(Q=1)=q$, for some $q \in (0,0.5)$.  
Consider $q=p=0.1$ for a numerical evaluation. Let us first consider the inner bound $\mathcal{R}_u(p_{\underline{X},Y})$ to the rate region $\mathcal{R}(p_{\underline{X},Y})$ given in 
\cite{atif2020ISITsource}, developed using unstructured code ensemble. Due to symmetry in the example, it turns out that the search over the auxiliary random variables for minimization reduces to a single-parameter minimization which can be computed through derivative techniques. The computation details are not provided for the sake of brevity. In particular, the minimum value of $R_1+R_2+C$ can be computed to be $1.3965$.  Next let us consider the new inner bound $\mathcal{R}_s(p_{\underline{X},Y})$ developed using structured code ensemble (Theorem 1). The minimum value of 
$R_1+R_2+C$ can be computed to be $0.9596$.

The results can also be verified for the special case of $q=0$ which we provide in the following.
Using the arguments given in proof of Proposition 1 of \cite{korner1979encode}, one can show 
that
\begin{align*}
\mathcal{R}_u(p_{\underline{X},Y}) &= 
\{(R_1,R_2,C): R_1 \geq h_b(p), R_2 \geq h_b(p), \\
& R_1+R_2 \geq 1+h_b(p), C\geq 0\}. 
\end{align*}
Next let us consider the new inner bound $\mathcal{R}_s(p_{\underline{X},Y})$ developed using structured code ensemble
(Theorem 1). 
By choosing $W_1=X_1$ and $W_2=X_2$, we see that the following triple of rates is achievable:
\[
\{(R_1,R_2,C): R_1 \geq h_b(p), R_2 \geq h_b(p), C\geq 0\}.
\]
In fact, one can show that this is optimal 
using the side information argument. If 
$X_2$ is sent losslessly, then from the converse argument in the side information case, we see that $R_1 \geq H(X_2|X_1)=h_b(p)$.

\end{example}

\section{Proof of Distributed Soft Covering using Algebraic Structured Random Codes}
The coding strategy used here is based on Unionized Coset Codes, defined in Definition \eqref{def:UCC}. The structure in these codes provides a method to exploit the structure present in the stochastic processing applied by decoder, i.e., $ P_{Y|W_1+W_2} $. Using this technique, we aim to strictly reduce the rate constraints compared to the ones obtained in Theorem 1 of \cite{atif2020ISITsource}.

Let $\mu \in [2^{nC}]$  denote the common randomness shared amidst 
all terminals. The first encoder uses a part of the entire common randomness available to it, say $C_1$ bits out of the $C$ bits, which is denoted by $\mu_1 \in [2^{nC_1}]$. Similarly, let $\mu_2 \in [2^{nC_2}]$ denote the common randomness used by the second encoder.  Our goal is to prove the existence of PMFs 
$p_{M_1|X_1^n}^{(\mu_1)}(m_1|x_1^n) : x_1^n \in \mathcal{X}_1^n, m_1 \in 
[\Theta_1], \mu_1 \in [2^{nC_1}], p^{\mu_2}_{M_2|X_2^n}(m_2|x_2^n) : x_2^n \in 
\mathcal{X}_2^n, m_2 \in [\Theta_2], \mu_2 \in [2^{nC_2}]$, $p_{Y^n|M_1,M_2}(y^n| m_1,m_2) : y^n \in 
\mathcal{Y}^n, (m_1,m_2) \in 
[\Theta_1] \times [\Theta_2]$ such that 
\begin{align}
 \mathscr{Q}\define \!\cfrac{1}{2}&\!\!\sum_{\underline{x}^n,y^n}\Bigg|p^n_{\underline{X}Y}(\underline{x}^n,y^n) - \!\!\!\!\sum_{\substack{ \mu \in
[2^{nC}]  }}\sum_{\substack{m_1 \in[ \Theta_1],\\m_2 \in[ \Theta_2]}} \!\!\!\!
\frac{p^{n}_{\underline{X}}(\underline{x}^n)}{2^{nC}} p^{(\mu_1)}_{M_1|X_1^{n}}(m_1|x_1^{n})p^{(\mu_2)}_{M_2|X_2^{n}}(m_2|x_2^{n}
)p^{(\mu) } _ { Y^ { n } 
|\underline{M}}(y^{n}|\underline{m})\Bigg| \leq \varepsilon, \nonumber\\
\label{eq:dist_main_lemma}
&\frac{\log \Theta_j}{n} \leq R_j + \epsilon : j \in [2],&
\end{align}
for sufficiently large $n$. Fix a block length $n>0$, a positive integer $N$ and a finite field $\FF_p$. Further, let $W_1$ and $ W_2 $ be random variables defined on the alphabets $ \mathcal{W}_1 $ and $ \mathcal{W}_2 $, respectively, where $ \mathcal{W}_1=\mathcal{W}_2=\FF_p, $ and let $Z \deq W_1 \oplus W_2.$ In building the code, we use the Unionized Coset Codes (UCCs)  \cite{pradhanalgebraic} defined as below. These codes involve two layers of codes (i) a coarse code and (ii) a fine code. The coarse code is a coset of the linear code and the fine code is the union of several cosets of the linear code.

For a fixed $k \times n$ matrix $G \in \mathbb{F}_p^{k \times n}$ with 
$k \leq n$, and a $1 \times n$
vector $B \in \mathbb{F}_p^n$, define the coset code as 
\[
\mathbb{C}(G,B)\deq  \{ x^n: x^n = a^{k} G+B, \mbox{ for some } a^{k} \in \mathbb{F}_p^{k}
\}.
\]
In other words, $\mathbb{C}(G,B)$ is a shift of the row space of
the matrix $G$. The row space of $G$ is a linear code. 
If the rank of $G$ is $k$, then there are $p^{k}$ codewords in the coset
code.

\begin{definition}\label{def:UCC}
An $(n,k,l,p)$ UCC is a pair $(G,h)$ consisting of a $k\times n$ matrix $G \in \mathbb{F}_p^{k \times n}$, and a mapping $h:\mathbb{F}_p^{l}
\rightarrow \mathbb{F}_p^n$. In the context of UCC, define the composite code as
$\mathbb{C}=\bigcup_{i \in \mathbb{F}_p^{l}} \mathbb{C}(G,h(i))$.
\end{definition}

For every $ \mu \deq (\mu_{1},\mu_{2}) $, consider two UCCs $ (G^{},h_1^{(\mu_1)}) $ and $ (G^{},h_2^{(\mu_2)}) $, each with parameters $ (n,k,l_1,p) $ and $ (n,k,l_2,p) $, respectively. Note that, for every $ \mu\in [N], $ the generator matrix $ G^{}$ remains the same.

For each $ (\mu_1,\mu_2) $, the generator matrix $ G^{} $ along with the function $ h_1^{\mu_1} $ and $ h_2^{\mu_2} $ generates $ p^{k + l_1} $ and $ p^{k+l_2} $ codewords, respectively. Each of these codewords are characterized by a triple $ (a_i,m_i,\mu_i)$, where $ a_i \in \FF^{k}_p $ and $ m_i \in \FF^{l_i}_p$  corresponds to the coarse code and the fine code indices, respectively, for $ i\in[2]$. Let $ \UCodeword $ and $ \VCodeword $ denote the codewords associated with   Alice and Bob,  generated using the above procedure, respectively, where $
\UCodeword \define a_1 G^{} + h_1^{(\mu_1)}(i),$ and  $\VCodeword \define a_2 G^{} + h_2^{(\mu_2)}(j).$

Consider the collections $c_{1} = 
(c_{1}^{(\mu_1)}:1 \leq \mu_1 \leq 2^{nC_1})$ where $c_{1}^{(\mu_1)} = 
(\mathtt{w}_{1}(l_{1},\mu_1) : 1\leq l_{1} \leq 2^{n\tilde{R}_{1}})$ and $c_2 = 
(c_{1}^{(\mu_1)}:1 \leq \mu_1 \leq 2^{nC_1})$
where $c_2^{(\mu_2)}
= (\mathtt{w}_{2}(l_{2},\mu_2): 1 \leq l_{2} \leq 
2^{n\tilde{R}_{2}})$. For this 
collection, we let 
\begin{align}
    E^{(\mu_1)}_{L_1|X_1^n}(a_1,\!m_1|x_1^n) &\deq \sum_{\substack{w_1^n \in T_{\delta}(W_1|x_1^n)}}{p^n}\frac{p^n_{W_1|X_1}\!\!(w_1^n|x_1^n)}{2^{nS_1}(1+\eta)}\!\mathbbm{1}_{\{\UCodeword = w_1^n\}},\nonumber \\
    E_{L_2|X_2^n}^{(\mu_2)}(a_2,\!m_2|x_2^n) &\deq \sum_{\substack{w_2^n \in T_{\delta}(W_2|x^n_2)}}{p^n}\frac{p^n_{W_2|X_2}\!(w_2^n|x_2^n)}{2^{nS_2}(1+\eta)}\!\mathbbm{1}_{\{\VCodeword = w_2^n\}}.\label{def:E_L|X}
\end{align}
The definition of $E^{(\mu_1)}_{L_1|X_1^n}$ and $E^{(\mu_2)}_{L_2|X_2^n}$ can be thought of as encoding rules that do not exploit the additional rebate obtained by using binning techniques, specifically in a distributed setup.

\subsection{Binning of Random Encoders}
\label{sec:PMF binning}
We next proceed to binning the above constructed collection of random encoders.
Since, UCC is already a union of several cosets, we associate a bin to each coset, and place all the codewords of a coset in the same bin. For each $i\in \FF_p^{l_1}$ and $j\in \FF_p^{l_2}$, let $\mathcal{B}^{(\mu_1)}_1(i) \deq \mathbb{C}(G^{},h_1^{(\mu_1)}(i))$ and $\mathcal{B}^{(\mu_2)}_2(j) \deq \mathbb{C}(G^{},h_2^{(\mu_2)}(j))$ denote the $i^{th}$ and the  $j^{th}$ bins, respectively. Formally, we define the following PMFs. 
\begin{eqnarray}\label{Eqn:dist_PMFInducedByEncoder1}
p_{M_{1}|X_{1}^{n}}^{(\mu_1)}(m_{1}|x_{1}^{n}) =
\begin{cases}
\mathbbm{1}_{\{m_i=0\}}& \mbox{if } s_{i}^{(\mu_i)}(x_{i}^{n}) > 1,\\
1-s_{i}^{(\mu_i)}(x_{i}^{n})&\mbox{if }m_{i}=0 \mbox{ and } s_{i}^{(\mu_i)}(x_{i}^{n}) \!\in [0, 1],\\
\displaystyle \sum_{a_i\in\FF_p^{k}}E^{(\mu_i)}_{L_i|X_i^n}(a_i,\!m_i|x_i^n) &\mbox{if }m_{i}\neq0 \mbox{ and 
}s_{i}^{(\mu_i)}(x_{i}^{n}) \!\in [0, 1],
\end{cases}\nonumber
\end{eqnarray}
for all $x_i^n \in T_{\delta}(X_i)$, $s_i^{(\mu_i)}(x_{i}^{n})$ defined as $s_i^{(\mu_i)}(x_{i}^{n}) \deq \sum_{a_i\in\FF_p^{k}}\sum_{m_i \in \FF_p^{l_i}}E^{(\mu_i)}_{L_i|X_i^n}(a_i,m_i|x_i^n)$ and $i\in[2]$.
For $x_1^n \notin T_{\delta}(X_1)$, we let $ p_{M_{1}|X_{1}^{n}}^{(\mu_1)}(m_{1}|x_{1}^{n}) = \mathbbm{1}_{\{m_1=0\}}.$

With this definition note that, $ \sum_{m_1=0}^{2^{nR_1}}p_{M_1|X_1^n}^{(\mu_1)}(m_1|x_1^n) = 1$ for all $\mu_1 \in [2^{nC_1}]$ and $x^n_1 \in \mathcal{X}_1^n$ and similarly, $ \sum_{m_2=0}^{2^{nR_2}}p_{M_2|X_2^n}^{(\mu_2)}(m_2|x_2^n) = 1$ for all $\mu_2 \in [2^{nC_2}]$ and $ x^n_2 \in \mathcal{X}_2^n$. 

Also, note that the effect of introducing binning (by defining the above PMFs) is in reducing the communication rates from $(S_1, S_2)$ to $(R_1,R_2)$, where $ R_i = \frac{l_i}{n}\log{p}, i\in\{1,2\}$. Now, we move on to describing the decoder.

\subsection{Decoder mapping }
We create a decoder that takes as an input a pair of bin numbers and produces a sequence  $W^n \in \FF^n_p$. More precisely, we define a mapping $f^{(\mu)}$ for $\mu \deq (\mu_1,\mu_2)$, acting on the messages $(m_1,m_2)$ as follows. 
On observing $ \mu $ and the classical indices $ (m_1,m_2) \in \FF_p^{l_1}\times \FF_p^{l_2} $ communicated by the encoder, the decoder constructs $D^{(\mu)}_{i,j} \deq \{\tilde{a} \in \FF_p^{k}:   \tilde{a}G^{} + h_1^{(\mu_1)}(i)+h_2^{(\mu_2)}(j) \in \mathcal{T}_{\hat{\delta}}^{(n)}(Z) \},$ and $f^{(\mu)}(m_1,m_2)$
\begin{align}
 \deq
 \begin{cases}
	  \tilde{a}G^{} + h_1^{(\mu_1)}(i)+h_2^{(\mu_2)}(j) &\quad \text{ if  } D^{(\mu_1,\mu_2)}_{i,j} \equiv \{\tilde{a}\} \\
	w^n_0 &\quad \text{ otherwise  }, 
\end{cases}
\label{eq:POVM-17}
\end{align}
where $\hat{\delta} = p\delta$ and
{$w_0^n$ is an additional sequence added to $ \FF_p^n $}.
Further, $f^{(\mu)}(m_1,m_2)=w_0^n$ for $i=0$ or $j=0$.
The decoder then performs a stochastic processing of the output and chooses $y^n$ according to PMF $p^n_{Y|Z}(y^n|f^{(\mu)}(m_1,m_2))$. This 
implies the PMF $p^{(\mu_1)}_{Y^n|M_1 M_2}(\cdot|\cdot)$ is given by
\begin{eqnarray}
\label{Eqn:dist_PMFInducedByDecoder}
p^{(\mu)}_{Y^n|M_1M_2}(\cdot|m_1,m_2) = p^n_{Y|Z}(y^{n}|f^{(\mu)}(m_1,m_2)).
\end{eqnarray}
We now begin our analysis of the total variation term given in (\ref{eq:dist_main_lemma}). 
\subsection{Analysis of Total Variation}
Our goal is to prove the existence of a collections $c_{1},c_{2}$ for which (\ref{eq:dist_main_lemma}) holds. We do this via random coding. Specifically, we prove that $\EE[{K}] \leq \epsilon$, where the expectation is over the ensemble of codebooks. The PMF induced on the ensemble of codebooks is as specified below. The codewords of the random codebook  
$C_{i}^{(\mu_i)} = (\mathtt{W}_i(a_i,m_i,\mu_i): a_i \in \FF_p^k, m_i \in \FF_p^{l_i})$ for each $\mu_i \in [2^{nC_i}]$ 
are only pairwise independent \cite{pradhanalgebraic} and distributed with PMF $ \mathbb{P}(\mathtt{W}_{i}(a_i,m_i,\mu_i ) = w_{i}^{n}) = \frac{1}{p^n}$
for each $i \in [2].$

\noindent \textbf{Step 1: Error caused by not covering}\\ 
We begin by splitting $K$ into two terms using the triangle inequality as $K \leq S + \tilde{S}$, where
\begin{align}
    S \deq &\sum_{\underline{x}^n,y^n}\Bigg|p^n_{\underline{X}Y}(\underline{x}^n,y^n) - \sum_{\substack{ \mu_1,\mu_2  }}\sum_{\substack{m_1 >0,\\m_2 >0}} 
\frac{p^{n}_{\underline{X}}(\underline{x}^n)p^{(\mu_1)}_{M_1|X_1^{n}}(m_1|x_1^{n})}{2^{n(C_1+C_2)}} p^{(\mu_2)}_{M_2|X_2^{n}}(m_2|x_2^{n}
)p^{(\mu) } _ { Y^ { n } 
|\underline{M}}(y^{n}|\underline{m})\Bigg|, \nonumber\\
\tilde{S} \deq &\sum_{\underline{x}^n,y^n}\Bigg|\sum_{\substack{ \mu_1,\mu_2 }}\sum_{\substack{m_1=0 \medcup m_2=0}}
\frac{p^{n}_{\underline{X}}(\underline{x}^n)}{2^{n(C_1+C_2)}}p^{(\mu_1)}_{M_1|X_1^{n}}(m_1|x_1^{n})p^{(\mu_2)}_{M_2|X_2^{n}}(m_2|x_2^{n}
)p^{(\mu) } _ { Y^ { n } 
|\underline{M}}(y^{n}|\underline{m})\Bigg|. \nonumber
\end{align}

\noindent Note that $\widetilde{S}$ captures the error induced by not covering $p^n_{\underline{X}Y}$.
For the term corresponding to $\widetilde{S}$, we 
prove the following result by developing the following lemma below followed by a proposition. 

\begin{lem}\label{lem:NotaPMFBound}
For the above defined notations, for $i \in \{1,2\}$, if $S_i \geq I(X_i;W_i) + \delta_{c_i}$, then the following holds true
\begin{align}\label{eq:NotaPMF}
\frac{1}{2^{nC_i}}\sum_{\substack{ \mu_i  }}\sum_{{x}_i^n}{p^{n}_{{X}_i}({x}_i^n)} \PP\Bigg(\bigg[\sum_{a_i\in\FF_p^{k}}\sum_{m_i \in \FF_p^{l_i}}E^{(\mu_i)}_{L_i|X_i^n}(a_i,m_i|x_i^n)\bigg] > 1\Bigg) \leq  \epsilon_{c_i}
\end{align}
\end{lem}
\begin{proof}
 The proof is provided in Appendix \ref{proof_lem:NotaPMFBound}.
\end{proof}
\begin{proposition}\label{prop:Lemma for Tilde_S}
There exist functions  $\epsilon_{\widetilde{S}}(\delta), $ and $\delta_{\widetilde{S}}(\delta)$, 
such that for  all sufficiently small $\delta$ and sufficiently large $n$, we have $\EE[\widetilde{S}] \leq\epsilon_{\widetilde{S}}(\delta) $, if  $S_1 >  I(X_1;W_1) - H(W_1)  +\log{p} + \delta_{\widetilde{S}} $ and $S_2 > I(X_2;W_2) - H(W_2) + \log{p} +\delta_{\widetilde{S}},$ where $\epsilon_{\widetilde{S}}, \delta_{\widetilde{S}} \searrow 0$ as $\delta \searrow 0$. 
\end{proposition}
\begin{proof}
The proof is provided in Appendix \ref{proof_prop:Lemma for Tilde_S}
\end{proof}

Now we move on to removing from $S$ the error that is induced due to binning.

\noindent{\bf Step 2: Error caused by binning}\\
Note that $S$ can be simplified using the definitions of $P_{M_1|X_1^n}^{(\mu_1)}(\cdot|\cdot)$, $P_{M_2|X_2^n}^{(\mu_2)}(\cdot|\cdot)$, and $p^{(\mu) } _ { Y^ { n } 
|\underline{M}}(y^{n}|\underline{m}) $ as
\begin{align}
    S \deq & \sum_{\underline{x}^n,y^n}\Bigg|p^n_{\underline{X}Y}(\underline{x}^n,y^n) - 
    \sum_{\substack{ \mu_1,\mu_2  }}\sum_{\substack{ m_1 > 0, \\m_2>0  }}\sum_{\substack{a_1 \in \FF_p^k\\a_2 \in \FF_p^k}}\sum_{\substack{w_1,w_2 \in \FF_p^n}} \frac{p^{n}_{\underline{X}}(\underline{x}^n)}{2^{n(C_1+C_2)}}E^{(\mu_1)}_{W^n_1|X_1^{n}}(w^n_1|x_1^{n})E^{(\mu_2)}_{W^n_2|X_2^{n}}\!(w^n_2|x_2^{n}) \nonumber \\ & \hspace{150pt}
     \mathbbm{1}_{\left\{ \UCodeword = w_1^n\right\}}\mathbbm{1}_{\left\{ \VCodeword = w_1^n\right\}}
p_{Y|Z}^n(y^{n}|f^{(\mu)}(m_1,m_2))\Bigg|,\nonumber
\end{align}
where  $E^{(\mu_i)}_{W_i^n|X_i^n}(w_i^n|x_i^n)$ is defined as 
\begin{align}
    E^{(\mu_i)}_{W_i^n|X_i^n}(w_i^n|x_i^n)  & \deq  \frac{p^n}{2^{nS_i}(1+\eta)}{p_{W_i|X_i}^n(w_i^n|x_i^n)} \mathbbm{1}_{\left\{w_i^n \in T_{\delta}(W_i|x_i^n)\right\}} \11_{\{s_i^{(\mu_i)}(x_i^n) \leq 1\}}, \quad \text{ for i $\in \{1,2\}$ } .
    \nonumber
\end{align} 
Further by defining $\gammaCoeff$ and $\zetaCoeff$ as 
\begin{align}
    \gammaCoeff & \deq |\{\UCodeword = w_1^n\} | = \sum_{m_1 > 0}\sum_{a_1 \in \FF_{q}^k} \11_{\{\UCodeword = w_1^n\}} \quad \text{ and } \nonumber \\
    \zetaCoeff & \deq |\{\VCodeword = w_2^n\} | = \sum_{m_2 > 0}\sum_{a_2 \in \FF_{q}^k} \11_{\{\VCodeword = w_2^n\}}.
\end{align}we bound $S$ using triangle inequality as $S \leq S_1 + S_2$, where
\begin{align}
    S_1 &\deq \sum_{\underline{x}^n,y^n}\Bigg|p^n_{\underline{X}Y}(\underline{x}^n,y^n) - \sum_{\substack{ \mu_1,\mu_2  }}\sum_{\substack{w_1^n,w_2^n \in \FF_p^n}} \frac{p^{n}_{\underline{X}}(\underline{x}^n)}{2^{n(C_1+C_2)}} \gammaCoeff \zetaCoeff
    \nonumber \\
     &\hspace{170pt}E^{(\mu_1)}_{W^n_1|X_1^{n}}(w^n_1|x_1^{n})E^{(\mu_2)}_{W^n_2|X_2^{n}}(w^n_2|x_2^{n}
)p^{n } _ {Y|Z}(y^{n}|w_1^n\oplus w_2^n)\Bigg|,\nonumber \\
S_2 &\deq \sum_{\underline{x}^n,y^n}
     \sum_{\substack{ \mu_1,\mu_2  }}\sum_{\substack{ m_1 > 0, \\m_2>0  }}\sum_{\substack{a_1 \in \FF_p^k\\a_2 \in \FF_p^k}} \sum_{\substack{w_1^n,w_2^n \in \FF_p^n}} \!\frac{p^{n}_{\underline{X}}(\underline{x}^n)}{2^{n(C_1+C_2)}} E^{(\mu_1)}_{W^n_1|X_1^{n}}(w^n_1|x_1^{n})E^{(\mu_2)}_{W^n_2|X_2^{n}}(w^n_2|x_2^{n}
) \nonumber \\ &\hspace{80pt} \mathbbm{1}_{\left\{ \UCodeword = w_1^n\right\}}\mathbbm{1}_{\left\{ \VCodeword = w_1^n\right\}} \Bigg|p^n_{Y|Z}(y^{n}|w_1^n\oplus w_2^n)\! -\! p^n_{Y|Z}(y^{n}|f^{(\mu)}(m_1,\!m_2))\Bigg|.\nonumber
\end{align}

To bound the term corresponding to $S_2$, we provide the following proposition.
\begin{proposition}[Mutual Packing]\label{prop:Lemma for S_2}
There exist  $\epsilon_{S_{2}}(\delta),$ such that for  all sufficiently small $\delta$ and sufficiently large $n$, we have $\EE\left[{S}_2\right] \leq \epsilon_{{S_2}}(\delta) $, if $ S_{1} - R_1 < \log{p} - H(Z)$, or equivalently, $ S_{2} - R_2 < \log{p} - H(Z)$,  where  $\epsilon_{{S}_2} \searrow 0$ as $\delta \searrow 0$.
\end{proposition}
\begin{proof}
The proof is provided in Appendix  \ref{proof_prop:Lemma for S_2}.
\end{proof}
\noindent Now, we move on to analyzing the term corresponding to $S_1$.\\
\noindent\textbf{Step 3: Term concerning Alice's encoding}\\
\noindent In this step, we separately analyze the action of the two encoders in approximating the product distribution $p_{\underline{X}Y}^n(\cdot)$. For that, we split $S_{1}$ as $ S_{1} \leq Q_1 + Q_2 $, where
\begin{align}
    Q_1 \deq &  \sum_{\underline{x}^n,y^n}\Bigg|p^n_{\underline{X}Y}(\underline{x}^n,y^n) - \frac{1}{2^{nC_1}}\sum_{\substack{ \mu_1 }}\sum_{\substack{w_1^n,w_2^n \in \FF_p^n}}{p^{n}_{\underline{X}}(\underline{x}^n)}\gammaCoeff E^{(\mu_1)}_{W^n_1|X_1^{n}}(w^n_1|x_1^{n})\nonumber \\ & \hspace{3in} p^n_{W^n_2|X_2^{n}}(w^n_2|x_2^{n}
)p^n_ { Y|{Z}}(y^{n}|w_1^n+w_2^n)\Bigg|,\nonumber \\
Q_2 \deq& \!\! \sum_{\underline{x}^n,y^n}\Bigg|\frac{1}{2^{nC_1}}\sum_{\substack{ \mu_1 }}\sum_{\substack{w_1^n,w_2^n \in \FF_p^n}} {p^{n}_{\underline{X}}(\underline{x}^n)}\gammaCoeff E^{(\mu_1)}_{W^n_1|X_1^{n}}(w^n_1|x_1^{n}) \nonumber \\ & \hspace{80pt}\left(p^n_{W^n_2|X_2^{n}}(w^n_2|x_2^{n})\! - \zetaCoeff E^{(\mu_2)}_{W^n_2|X_2^{n}}(w^n_2|x_2^{n})\right)
p^n_ { Y|{Z}}(y^{n}|w_1^n\!\!+\!\!w_2^n) \Bigg|.\nonumber 
\end{align}
With this partition, the terms within the trace norm of $ Q_1 $ differ only in the action of Alice's encoder. And similarly, the terms within the norm of $ Q_2 $ differ only in the action of Bob's encoder. Showing that these two terms are small forms a major portion of the achievability proof. \\
\noindent{\textbf{Analysis of $ Q_1$}:}  
To prove $ Q_1 $ is small, we characterize the rate constraints which ensure that an upper bound to $ Q_1 $ can be made to vanish in an expected sense. In addition, this upper bound becomes useful in obtaining a single-letter characterization for the rate needed to make the term corresponding to $ Q_2 $ vanish. For this, we define $ J $ as
\begin{align}
    J \deq &\sum_{\underline{x}^n,w_2^n,y^n}\Bigg|p^n_{\underline{X}W_2Y}(\underline{x}^n,y^n) - \frac{1}{2^{nC_1}}\sum_{\substack{ \mu_1}}\sum_{\substack{w_1^n}} {p^{n}_{\underline{X}}(\underline{x}^n)}\gammaCoeff E^{(\mu_1)}_{W^n_1|X_1^{n}}(w^n_1|x_1^{n})\nonumber \\
     &\hspace{3in}p^n_{W^n_2|X_2^{n}}(w^n_2|x_2^{n})
p^n_ { Y|{Z}}(y^{n}|w_1^n+w_2^n)\Bigg|.\nonumber 
\end{align}
By again using triangle inequality we obtain $J \leq J_1 + J_2$, where
\begin{align*}
    J_1 & \deq \sum_{\underline{x}^n,w_2^n,y^n}\Bigg|p^n_{\underline{X}W_2Y}(\underline{x}^n,y^n) - \frac{1}{2^{nC_1}}\sum_{\substack{ \mu_1}}\sum_{\substack{w_1^n}} {p^{n}_{\underline{X}}(\underline{x}^n)}\gammaCoeff \bar{E}^{(\mu_1)}_{W^n_1|X_1^{n}}(w^n_1|x_1^{n})\nonumber \\
     &\hspace{3in}p^n_{W^n_2|X_2^{n}}(w^n_2|x_2^{n})
p^n_ { Y|{Z}}(y^{n}|w_1^n+w_2^n)\Bigg| \\
    J_2 & \deq \sum_{\underline{x}^n,w_2^n,y^n}\Bigg|\frac{1}{2^{nC_1}}\sum_{\substack{ \mu_1}}\sum_{\substack{w_1^n}} {p^{n}_{\underline{X}}(\underline{x}^n)}\gammaCoeff \left(\bar{E}^{(\mu_1)}_{W^n_1|X_1^{n}}(w^n_1|x_1^{n}) - E^{(\mu_1)}_{W^n_1|X_1^{n}}(w^n_1|x_1^{n})\right)\nonumber \\
     &\hspace{3in}p^n_{W^n_2|X_2^{n}}(w^n_2|x_2^{n})
p^n_ { Y|{Z}}(y^{n}|w_1^n+w_2^n)\Bigg|
\end{align*}
where $\bar{E}^{(\mu_1)}_{W^n_1|X_1^{n}}(\cdot|\cdot)$ is defined as
\begin{align}
    E^{(\mu_1)}_{W_1^n|X_1^n}(w_1^n|x_1^n)  & \deq  \frac{p^n}{2^{nS_1}(1+\eta)}{p_{W_1|X_1}^n(w_1^n|x_1^n)} \mathbbm{1}_{\left\{w_1^n \in T_{\delta}(W_1|x_1^n)\right\}}.
\end{align}
To prove the term corresponding to $J_1$ is small, consider the following proposition.
\begin{proposition}
\label{prop:Lemma for J1}
There exist  $\epsilon_{J_1}(\delta), \delta_{J_1}(\delta)$ such that for  all sufficiently small $\delta$ and sufficiently large $n$, we have $\EE\left[J_1\right] \leq \epsilon_{J_1}$ if $S_1 + C_1 \geq I(W_1;X_1X_2ZW_2) +\log{p} - H(W_1) + \delta_{J_1}$, where  $\epsilon_{{J_1}}, \delta_{{J_1}} \searrow 0$ as $\delta \searrow 0$.
\end{proposition}
\begin{proof}
    The proof is provided in Appendix \ref{proof_prop:Lemma for J1}.
\end{proof}
\noindent Now, consider the term corresponding to $J_2$. This can be simplified as
\begin{align*}
    J_2 &= \frac{1}{2^{nC_1}}\sum_{{x}_1^n} \sum_{\substack{ \mu_1}} {p^{n}_{\underline{X}}(\underline{x}^n)}\left(\sum_{\substack{w_1^n}}\gammaCoeff \bar{E}^{(\mu_1)}_{W^n_1|X_1^{n}}(w^n_1|x_1^{n})\right)\11_{\{s_1^{(\mu_1)}(x^n_1) > 1\}} \\
    & \leq \frac{1}{2^{nC_1}}\sum_{{x}_1^n} \sum_{\substack{ \mu_1}} {p^{n}_{\underline{X}}(\underline{x}^n)} \EE\left[\sum_{\substack{w_1^n}}\gammaCoeff \bar{E}^{(\mu_1)}_{W^n_1|X_1^{n}}(w^n_1|x_1^{n})\right]\11_{\{s_1^{(\mu_1)}(x^n_1) > 1\}}  \\
    & \hspace{10pt} +  \frac{1}{2^{nC_1}}\sum_{{x}_1^n} \sum_{\substack{ \mu_1}} {p^{n}_{\underline{X}}(\underline{x}^n)} \left|\sum_{\substack{w_1^n}}\gammaCoeff \bar{E}^{(\mu_1)}_{W^n_1|X_1^{n}}(w^n_1|x_1^{n}) - \EE\left[\sum_{\substack{w_1^n}}\gammaCoeff \bar{E}^{(\mu_1)}_{W^n_1|X_1^{n}}(w^n_1|x_1^{n})\right]\right|  \\
    & \leq \frac{1}{2^{nC_1}}\sum_{{x}_1^n} \sum_{\substack{ \mu_1}} {p^{n}_{\underline{X}}(\underline{x}^n)}\EE\left[\sum_{\substack{w_1^n}}\gammaCoeff \bar{E}^{(\mu_1)}_{W^n_1|X_1^{n}}(w^n_1|x_1^{n})\right] \11_{\{s_1^{(\mu_1)}(x^n_1) > 1\}}  \\
     & \hspace{10pt} +  \frac{1}{2^{nC_1}}\sum_{{x}_1^n} \sum_{\substack{ \mu_1}} {p^{n}_{\underline{X}}(\underline{x}^n)} \left|\sum_{\substack{w_1^n}}\gammaCoeff \bar{E}^{(\mu_1)}_{W^n_1|X_1^{n}}(w^n_1|x_1^{n}) - \EE\left[\sum_{\substack{w_1^n}}\gammaCoeff \bar{E}^{(\mu_1)}_{W^n_1|X_1^{n}}(w^n_1|x_1^{n})\right]\right| \\
    & \leq H_0 + H_1 + \epsilon'',
    \end{align*}
   for
    \begin{align*}
    H_0 & \deq \frac{1}{2^{nC_1}}\sum_{{x}_1^n} \sum_{\substack{ \mu_1}} {p^{n}_{\underline{X}}(\underline{x}^n)} \11_{\{s_1^{(\mu_1)}(x^n_1) > 1\}}  \\
    H_1 & \deq  \frac{1}{2^{nC_1}}\sum_{{x}_1^n} \sum_{\substack{ \mu_1}}  \left| {p^{n}_{\underline{X}}(\underline{x}^n)} - \frac{p^n}{2^{nS_1}(1+\eta)} \sum_{\substack{w_1^n}}\sum_{m_1 > 0}\sum_{a_1 \in \FF_{q}^k}  {p_{X_1W_1}^n(x_1^n,w_1^n)}\11_{\{\UCodeword = w_1^n\}}\right|
\end{align*}
where we use $ \EE\left[\sum_{\substack{w_1^n}}\gammaCoeff \bar{E}^{(\mu_1)}_{W^n_1|X_1^{n}}(w^n_1|x_1^{n})\right] \leq 1$ in defining $H_0$ and $H_1$ is obtained by adding the sequences $w_1^n \notin \TDelta(W_1)$ within the summation.
Now, we can provide an upper bound on $H_0$ and $H_1$ using the Lemmas \ref{lem:NotaPMFBound} and \ref{lem:changeMeasureCoveringLemma}, respectively, as $\EE[H_0 + H_1 ] \leq \epsilon_{H} $ if $S_1  \geq I(X_1;W_1) + \delta_{H}$. Therefore, since $Q_1 \leq J$, hence $ Q_1$, can be made arbitrarily small for sufficiently large n, if $S_1+C_1 > I(W_1;X_1X_2YW_2) - H(W_1) + \log{p} + \delta_{J}$ . Now we move on to bounding $ Q_2 $.\\
\noindent{\bf Step 4: Analysis of Bob's encoding}\\Step 3 ensured that the random variables $X_1X_2YW_2$ are close to a product PMF in total variation. In this step, we approximate the PMF of random variables $X_1X_2Y$ using the Bob's encoding rule and bound the theorem corresponding to $Q_2$. We proceed with the following proposition. 
\begin{proposition} \label{prop:Lemma for Q2}
There exist functions  $\epsilon_{Q_{2}}(\delta) $ and $\delta_{{Q}_{2}}(\delta)$, such that for  all sufficiently small $\delta$ and sufficiently large $n$, we have $\EE[{Q}_2] \leq\epsilon_{{Q_2}}$, if $ S_1+C_1\geq  I(W_1;X_1X_2YW_2) - H(W_1)_{} + \log{p} + \delta_{Q_2}$ and $ S_2+C_2 \geq  I(W_2;X_1X_2Y) - H(W_2)_{} + \log{p} + \delta_{Q_2}$, where $\epsilon_{Q_2},\delta_{{Q}_2}  \searrow 0$ as $\delta \searrow 0$.
\end{proposition}
\begin{proof}
 The proof is provided in Appendix \ref{proof_prop:Lemma for Q_2}.
\end{proof}
Hence, in bounding the terms corresponding to $Q_1$ and $Q_2$, we have obtained the following constraints: 
\begin{align}
    S_1 + C_1 &\geq I(W_1;X_1X_2YW_2) - H(W_1)_{} + \log{p}, \nonumber \\
    S_2+C_2 &\geq  I(W_2;X_1X_2Y) - H(W_2)_{} + \log{p}. \label{SCrate1}
\end{align}
By doing an exact symmetric analysis, but by replacing the first encoder by a product distribution instead of the second encoder in $S_1$, we obtain the following constraints
\begin{align}
    S_1 + C_1 &\geq I(W_1;X_1X_2Y) - H(W_1)_{} + \log{p}, \nonumber \\
    S_2+C_2 &\geq  I(W_2;X_1X_2YW_1) - H(W_2)_{} + \log{p}. \label{SCrate2}
\end{align}
By time sharing between the above rates \eqref{SCrate1} and \eqref{SCrate2}, one can obtain the following rate constraints
\begin{align}
    S_1\! +\! C_1 & \!\geq\! I(W_1;X_1X_2Y) - H(W_1)_{} + \log{p}, \nonumber\\
    S_2\!+\!C_2 &\!\geq \! I(W_2;X_1X_2Y) - H(W_2)_{} + \log{p}, \nonumber\\
    S_1 \!+\!S_2\! +\!C_1\! +\!C_2 &\!\geq\! I(W_1W_2;X_1X_2Y)\! -\! H(W_1,W_2)\! + \!2\log{p}.  \nonumber
\end{align}

\subsection{Rate Constraints}

To sum-up, we showed that the  \eqref{eq:dist_main_lemma} holds for sufficiently large $n$ and with probability sufficiently close to 1, if the following bounds holds while incorporating the time sharing random variable $Q$ taking values over the finite set $\mathcal{Q}$\footnote{Since $Q$, the time sharing random variable is employed in the standard way we omit its discussion here.}:
\begin{align}\label{eq:rate-region_long}
    S_1 & \geq I(X_1;W_1|Q) - H(W_1|Q) + \log{p}, \nonumber \\
    S_2 & \geq I(X_2;W_2|Q) - H(W_2|Q) + \log{p}, \nonumber \\
    S_1+ C_1 & \geq I(X_1X_2Y;W_1|Q) - H(W_1|Q) + \log{p}, \nonumber \\
    S_2+ C_2 & \geq I(X_1X_2Y;W_2|Q) - H(W_2|Q) \!+\! \log{p}, \nonumber \\
    S_1 +S_2 +C_1 +C_2 &\geq I(W_1W_2;X_1X_2Y|Q) - H(W_1,\!W_2|Q)  + 2\log{p}, \nonumber \\
    S_1 - R_1 &= S_2 - R_2\leq \log{p} - H(W_1\oplus W_2|Q), \nonumber \\
    0 \leq R_1 &\leq S_1, \quad 0 \leq R_2 \leq S_2,\nonumber \\
    C_1+C_2&\leq C, \quad C \geq 0 
\end{align}

 Lastly,  we complete the proof of the theorem using the following lemma.
 \begin{lem}
 	Let $\mathcal{R}_1$ denote the set of all $(R_1,R_2,C)$ for which there exists $(S_1,S_2)$ such that the septuple $(R_1,R_2, C, S_1, S_2, {C}_1, {C}_2)$ satisfies the inequalities in \eqref{eq:rate-region_long}. Let, $\mathcal{R}_2$ denote the set of all triples $(R_1, R_2, C)$ that satisfies  the inequalities in \eqref{eq:rate-region} given in the statement of the theorem. Then, $
 	\mathcal{R}_1=\mathcal{R}_2$.
 \end{lem}
 \begin{proof}
 	This follows from Fourier-Motzkin elimination \cite{fourier-motzkin}.
 \end{proof}
 	


\appendices
\section{Proof of Lemmas}

\subsection{Proof of Lemma \ref{lem:NotaPMFBound}}\label{proof_lem:NotaPMFBound}
 Let $K$ denote the left hand side of \eqref{eq:NotaPMF}. Further, for the purpose of this proof, we skip the subscript $i$
Bounding the a-typical sequences of $x^n$ from the summation gives $K = K_1 + \epsilon_{X},$ where
\begin{align*}
    K_1 \deq \frac{1}{2^{nC}}\sum_{\substack{ \mu  }}\sum_{{x}^n \in \TDelta(X)}{p^{n}_{{X}}({x}^n)} \PP\Bigg(\bigg[\sum_{a\in\FF_p^{k}}\sum_{m \in \FF_p^{l}}E^{(\mu)}_{L|X^n}(a,m|x^n)\bigg] > 1\Bigg),
\end{align*}
and $\epsilon_{X} \deq \sum_{x^n \notin \TDelta(X)}{p^{n}_{{X}}({x}^n)}$. Note that $\epsilon_{X}(\delta) \searrow 0 $ as $\delta \searrow 0.$ With that, it remains to show the $K_1$ can be made arbitrarily small in expected sense.
Toward that, define
\begin{align*}
    Z^{(\mu)}_{x^n}(a,m) & \deq \frac{p^n}{(1+\eta)}\!\sum_{w^n }p_{XW}^n(x^n\!\!,w^n)\11_{\{\mathtt{w}^n(a,m,\mu) = w^n\}}\11_{\{w^n \in \TDelta(W|x^n)\}},\quad \!\!\!
    Z^{(\mu)}_{x^n}  \deq \frac{1}{2^{nS}}\!\!\sum_{m>0}\sum_{a \in \FF_q^k}\!\!Z^{(\mu)}_{x^n}(a,m)
\end{align*}
Observe the following  upper and lower bounds on $\EE[Z^{(\mu)}_{x^n}]$.
\begin{align}
    \EE[Z^{(\mu)}_{x^n}] & = \frac{p^n}{(1+\eta)}\sum_{w^n  \in \TDelta(W|x^n)}p_{XW}^n(x^n,w^n)\frac{1}{p^n} \leq \frac{p_X^n(x^n)}{(1+\eta)}  \label{eq:upperbound}\\
     \EE[Z^{(\mu)}_{x^n}] & = \frac{1}{(1+\eta)}\sum_{w^n  \in \TDelta(W|x^n)}p_{XW}^n(x^n,w^n) \geq \frac{p_X^n(x^n) 2^{n\delta_{w}}}{(1+\eta)} , \label{eq:lowerbound}
\end{align}
where the inequalities above uses the typicality arguments and $ \delta_{w}(\delta) \searrow 0$ as $\delta \searrow 0.$
Using these bounds, we perform the following simplification.
\begin{align}\label{eq:ineq1}
    \PP\Bigg(\bigg[\sum_{a\in\FF_p^{k}}\sum_{m \in \FF_p^{l}}E^{(\mu)}_{L|X^n}(a,m|x^n)\bigg] > 1\Bigg) & = \PP\Bigg(  Z^{(\mu)}_{x^n} > p_{X}^n(x^n) \Bigg)  \leq \PP\Bigg(  Z^{(\mu)}_{x^n} > (1+\eta)\EE[Z^{(\mu)}_{x^n} ] \Bigg) 
\end{align}
where the inequality above uses the inequality from \eqref{eq:upperbound}. Further, we have
\begin{align}\label{eq:ineq2}
      \PP\Bigg(  \bigg|Z^{(\mu)}_{x^n} - \EE[Z^{(\mu)}_{x^n}] \bigg| > \eta\EE[Z^{(\mu)}_{x^n} ] \Bigg) 
     \leq \frac{\EE\bigg|Z^{(\mu)}_{x^n} - \EE[Z^{(\mu)}_{x^n}] \bigg|}{\eta\EE[Z^{(\mu)}_{x^n}]} \leq \frac{(1+\eta)\sqrt{Var\left(Z^{(\mu)}_{x^n}\right)}}{\eta 2^{-n\delta_w}p^n_X(x^n)}
\end{align}
 where the first inequality follows from Markov Inequality, and the second uses (i) the Jensen's inequality for square-root function and (ii) the bound from \eqref{eq:lowerbound}.
 Combining the inequalities \eqref{eq:ineq1} and \eqref{eq:ineq2} using union bound, we obtain
 \begin{align}
    \PP\Bigg(\bigg[\sum_{a\in\FF_p^{k}}\sum_{m \in \FF_p^{l}}E^{(\mu)}_{L|X^n}(a,m|x^n)\bigg] > 1\Bigg) \leq \frac{\sqrt{Var\left(Z^{(\mu)}_{x^n}\right)}}{\eta 2^{-n\delta_w}p^n_X(x^n)} \leq \frac{(1+\eta)\sqrt{ 2^{-n(S + H(X|W) + H(W) -\delta'')}}}{\eta 2^{-n\delta_w}p^n_X(x^n)}  \label{eq:finalexp}
 \end{align}
where the last inequality follows by simplifying $Var\left(Z^{(\mu)}_{x^n}\right) $ similar to the one in \cite[Lemma 19]{cuff2009communication} to obtain
\begin{align}
    Var\left(Z^{(\mu)}_{x^n}\right) \leq 2^{-n(S + H(X|W) + H(W) -\delta'')}\frac{1}{(1+\eta)^2} \nonumber.
\end{align}
Substituting the simplification of \eqref{eq:finalexp} in $K_1$ completes the proof.

\section{Proof of Propositions} 
\subsection{Proof of Proposition \ref{prop:Lemma for Tilde_S}}
\label{proof_prop:Lemma for Tilde_S}
We bound $\Tilde{S}$ as $\Tilde{S} \leq \Tilde{S}_1 + \Tilde{S}_2 + \Tilde{S}_3, $ where
\begin{align}
    \Tilde{S}_1 & \deq \sum_{\underline{x}^n,y^n}\sum_{\substack{ \mu_1,\mu_2  }}\sum_{\substack{m_2 >0}} 
\frac{p^{n}_{\underline{X}}(\underline{x}^n)}{2^{n(C_1+C_2)}} p^{(\mu_1)}_{M_1|X_1^{n}}(0|x_1^{n}) p^{(\mu_2)}_{M_2|X_2^{n}}(m_2|x_2^{n})p^{(\mu)}_{Y^{n}|M_1M_2}(y^{n}|0,m_2) \nonumber \\ 
 \Tilde{S}_2 & \deq \sum_{\underline{x}^n,y^n}\sum_{\substack{ \mu_1,\mu_2  }}\sum_{\substack{m_2 >0}}
\frac{p^{n}_{\underline{X}}(\underline{x}^n)}{2^{n(C_1+C_2)}}p^{(\mu_1)}_{M_1|X_1^{n}}(m_1|x_1^{n}) p^{(\mu_2)}_{M_2|X_2^{n}}(0|x_2^{n})p^{(\mu)}_{Y^{n}|M_1M_2}(y^{n}|m_1,0) \nonumber \\ 
 \Tilde{S}_3 & \deq \sum_{\underline{x}^n,y^n}\sum_{\substack{ \mu_1,\mu_2  }}\sum_{\substack{m_2 >0}}
\frac{p^{n}_{\underline{X}}(\underline{x}^n)}{2^{n(C_1+C_2)}}p^{(\mu_1)}_{M_1|X_1^{n}}(0|x_1^{n}) p^{(\mu_2)}_{M_2|X_2^{n}}(0|x_2^{n})p^{(\mu)}_{Y^{n}|M_1M_2}(y^{n}|0,0) \nonumber 
\end{align}
\textbf{Analysis of $\tilde{S}_{1}$: } Consider the following simplification with regards to $\tilde{S}_1.$
\begin{align}
    \Tilde{S}_1 
&= \sum_{\underline{x}^n}\sum_{\substack{ \mu_1,\mu_2  }}
\frac{p^{n}_{\underline{X}}(\underline{x}^n)p^{(\mu_1)}_{M_1|X_1^{n}}(0|x_1^{n})}{2^{n(C_1+C_2)}}\left(\sum_{\substack{m_2 >0}}  p^{(\mu_2)}_{M_2|X_2^{n}}(m_2|x_2^{n})\left(\sum_{y^n}p^{(\mu)}_{Y^{n}|M_1M_2}(y^{n}|0,m_2)\right)\right) \nonumber \\ 
& = \sum_{{x}^n_1}\sum_{\substack{ \mu_1  }}
\frac{p^{n}_{{X}_1}({x}_1^n)p^{(\mu_1)}_{M_1|X_1^{n}}(0|x_1^{n})}{2^{nC_1}}  \nonumber \\  
 &\leq \sum_{{x}_1^n \in \TDelta(X_1)}\sum_{\substack{ \mu_1  }}
\frac{p^{n}_{{X}_1}({x}_1^n)p^{(\mu_1)}_{M_1|X_1^{n}}(0|x_1^{n})}{2^{nC_1}} \left(\11_{\{s_1^{(\mu_1)}(x^n_1) \leq 1\}} + \11_{\{s_1^{(\mu_1)}(x^n_1) > 1\}} \right) + \epsilon_{X_1} \nonumber \\
 & = \Tilde{S}_{11}  +\Tilde{S}_{12} + \epsilon_{X_1} \nonumber
\end{align}
where we define $\epsilon_{X_1}(\delta) \deq 1- \sum_{x^n \in \TDelta(X_1)}p^{n}_{{X}_1}({x}_1^n)$ and 
\begin{align}
    \Tilde{S}_{11} & \deq \sum_{{x}_1^n \in \TDelta(X_1)}\sum_{\substack{ \mu_1  }}\frac{p^{n}_{{X}_1}({x}_1^n)} {2^{nC_1}}(1-s_1^{(\mu_1)}(x^n_1))\11_{\{s_1^{(\mu_1)}(x^n_1) \leq 1\}}, \nonumber \\
    \Tilde{S}_{12} & \deq \sum_{{x}_1^n\in \TDelta(X_1)}\sum_{\substack{ \mu_1  }}\frac{p^{n}_{{X}_1}({x}_1^n)} {2^{nC_1}}\11_{\{s_1^{(\mu_1)}(x^n_1) > 1\}}. \nonumber 
\end{align}
Now, we bound each of the above terms.
For the term corresponding to $\tilde{S}_{11}$, consider the following.
\begin{align}
   \tilde{S}_{11} & \labelrel\leq{subeq:1_S11} \sum_{{x}_1^n \in \TDelta(X_1)}\sum_{\substack{ \mu_1  }}\frac{p^{n}_{{X}_1}({x}_1^n)} {2^{nC_1}}\left|1-\sum_{a_1\in\FF_p^{k}}\sum_{m_1 \in \FF_p^{l_1}}E^{(\mu_1)}_{L_1|X_1^n}(a_1,m_1|x_1^n)\right| \nonumber \\
    & \labelrel={subeq:2_S11}    \sum_{{x}_1^n \in \TDelta(X_1)}\sum_{\substack{ \mu_1  }}\frac{1} {2^{nC_1}}\left|p^{n}_{{X}_1}({x}_1^n)-\frac{p^n}{2^{nS_1}}\sum_{a_1\in\FF_p^{k}}\sum_{m_1 \in \FF_p^{l_1}}\sum_{\substack{w_1^n \in T_{\delta}(W_1|x_1^n)}}\frac{p^n_{W_1X_1}\!\!(w_1^n,x_1^n)}{(1+\eta)}\!\mathbbm{1}_{\{\UCodeword = w_1^n\}}\right| \nonumber \\
& \labelrel\leq{subeq:3_S11} \sum_{{x}_1^n \in \TDelta(X_1)}\sum_{\substack{ \mu_1  }}\frac{1} {2^{nC_1}}\left|p^{n}_{{X}_1}({x}_1^n)-\frac{p^n}{2^{nS_1}}\sum_{a_1\in\FF_p^{k}}\sum_{m_1 \in \FF_p^{l_1}}\sum_{\substack{w_1^n}}\frac{p^n_{W_1X_1}\!\!(w_1^n,x_1^n)}{(1+\eta)}\!\mathbbm{1}_{\{\UCodeword = w_1^n\}}\right| \nonumber \\ 
    & \hspace{30pt} + \label{eq:simplify_S11_tilde}
    \sum_{{x}_1^n \in \TDelta(X_1)}\sum_{\substack{ \mu_1  }}\frac{p^n}{2^{n(S_1+C_1)}}\sum_{a_1\in\FF_p^{k}}\sum_{m_1 \in \FF_p^{l_1}}\sum_{\substack{w_1^n \notin T_{\delta}(W_1|x_1^n)}}\frac{p^n_{W_1X_1}\!\!(w_1^n,x_1^n)}{(1+\eta)}\!\mathbbm{1}_{\{\UCodeword = w_1^n\}}
\end{align}
where \eqref{subeq:1_S11} follows by bounding the indicator by $1$ and using the definition of $s_1^{(\mu_1)}(\cdot)$, \eqref{subeq:2_S11} uses the definition of $E^{(\mu_1)}_{L_1|X_1^n}(a_1,\!m_1|x_1^n)$ as defined in Definition \eqref{def:E_L|X}, the inequality \eqref{subeq:3_S11} follows from triangle inequality. 

Taking expectation on \eqref{eq:simplify_S11_tilde} over the first encoders codebook generation, we obtain
\begin{align}
    \EE_{\mathbbm{C}_1}[\tilde{S}_{11}] & \leq \frac{1} {2^{nC_1}}\sum_{\substack{ \mu_1  }}\EE_{\mathbbm{C}_1}\left[\sum_{{x}_1^n \in\TDelta(X_1)} \left|p^{n}_{{X}_1}({x}_1^n)-\frac{p^n}{2^{nS_1}}\sum_{a_1\in\FF_p^{k}}\sum_{m_1 \in \FF_p^{l_1}}\sum_{\substack{w_1^n}}\frac{p^n_{W_1X_1}\!\!(w_1^n,x_1^n)}{(1+\eta)}\!\mathbbm{1}_{\{\UCodeword = w_1^n\}}\right|\right] \nonumber \\ 
    & \hspace{30pt} + \label{eq:simplify_S11_tilde1}
    \sum_{{x}_1^n \in \TDelta(X_1)}\sum_{\substack{ \mu_1  }}\frac{p^n}{2^{n(S_1+C_1)}}\sum_{a_1\in\FF_p^{k}}\sum_{m_1 \in \FF_p^{l_1}}\sum_{\substack{w_1^n \notin T_{\delta}(W_1|x_1^n)}}\frac{p^n_{W_1X_1}\!\!(w_1^n,x_1^n)}{(1+\eta)}\frac{1}{p^n}
\end{align}
For the first term in \eqref{eq:simplify_S11_tilde1},  we use Lemma \eqref{lem:changeMeasureCoveringLemma} and obtain $\EE[\tilde{S}_{11}] \leq \epsilon_{\tilde{S}_{11}}$ if $S_1 \leq I(X_1;W_1) - H(W_1) + \log{p} + \delta_{\tilde{S}_{11}}$. As for the second term we can use typicality arguments and bound it as
\begin{align}
    \sum_{{x}_1^n \in \TDelta(X_1)}& \sum_{\substack{ \mu_1  }}\frac{p^n}{2^{n(S_1+C_1)}}\sum_{a_1\in\FF_p^{k}}\sum_{m_1 \in \FF_p^{l_1}}\sum_{\substack{w_1^n \notin T_{\delta}(W_1|x_1^n)}}\frac{p^n_{W_1X_1}\!\!(w_1^n,x_1^n)}{(1+\eta)}\frac{1}{p^n} \nonumber \\
    & \leq \sum_{{x}_1^n \in \TDelta(X_1)}\sum_{\substack{w_1^n \notin T_{\delta}(W_1|x_1^n)}}\frac{p^n_{W_1X_1}\!\!(w_1^n,x_1^n)}{(1+\eta)} \leq {\epsilon'}
\end{align}
where $\epsilon' \deq \sum_{{x}_1^n \in \TDelta(X_1)}\sum_{\substack{w_1^n \notin T_{\delta}(W_1|x_1^n)}}{p^n_{W_1X_1}\!\!(w_1^n,x_1^n)} $

Finally, we have the term corresponding to  $\tilde{S}_{12}$. For this, we use Lemma \ref{lem:NotaPMFBound} and hence obtain $\EE[\tilde{S}_{12}] \leq \epsilon_{\tilde{S}_{12}}$ if $S_2 \leq I(X_1;W_1) - H(W_1) + \log{p} + \delta_{\tilde{S}_{12}}$.

\textbf{Analysis of $\tilde{S}_2$: } Due to the symmetry in $\tilde{S}_1$ and $\tilde{S}_2, $ the analysis of $\tilde{S}_2$ follows very similar arguments at that of $\tilde{S}_1$ and therefore we obtain $\EE[\tilde{S}_2] \leq \epsilon_{\tilde{S}_2}$ if $S_2 \geq I(X_2;W_2) - H(W_2) +\log{p} + \delta_{\tilde{S}_2}.$

\textbf{Analysis of $\tilde{S}_3: $} Follows by merging the above analysis of $\tilde{S}_1$ and $\tilde{S}_2.$

\subsection{Proof of Proposition \ref{prop:Lemma for S_2}} \label{proof_prop:Lemma for S_2}
Recalling $S_2$, we have
\begin{align}
    S_2 & = \sum_{\underline{x}^n}
     \sum_{\substack{ \mu_1,\mu_2  }}\sum_{\substack{ m_1 > 0, \\m_2>0  }}\sum_{\substack{a_1 \in \FF_p^k\\a_2 \in \FF_p^k}} \sum_{\substack{w_1^n,w_2^n \in \FF_p^n}} \!\frac{p^{n}_{\underline{X}}(\underline{x}^n)}{2^{n(C_1+C_2)}} E^{(\mu_1)}_{W^n_1|X_1^{n}}(w^n_1|x_1^{n})E^{(\mu_2)}_{W^n_2|X_2^{n}}(w^n_2|x_2^{n}
) \nonumber \\ &\hspace{50pt} \mathbbm{1}_{\left\{ \UCodeword = w_1^n\right\}}\mathbbm{1}_{\left\{ \VCodeword = w_1^n\right\}}
\sum_{y^n}\Bigg|p^n_{Y|Z}(y^{n}|w_1^n\oplus w_2^n)\! -\! p^n_{Y|Z}(y^{n}|f^{(\mu)}(m_1,\!m_2))\Bigg|\nonumber \\
& \leq 2 \sum_{\underline{x}^n}
     \sum_{\substack{ \mu_1,\mu_2  }}\sum_{\substack{ m_1 > 0, \\m_2>0  }}\sum_{\substack{a_1 \in \FF_p^k\\a_2 \in \FF_p^k}} \sum_{\substack{w_1^n,w_2^n \in \FF_p^n}} \!\frac{p^{n}_{\underline{X}}(\underline{x}^n)}{2^{n(C_1+C_2)}} \frac{p^np^n}{2^{n(S_1+S_2)}(1+\eta)^2}{p_{W_1|X_1}^n(w_1^n|x_1^n)} {p_{W_2|X_2}^n(w_2^n|x_2^n)}  \nonumber \\ &\hspace{55pt} \mathbbm{1}_{\left\{w_i^n \in T_{\delta}(W_i|x_i^n)\right\}}\mathbbm{1}_{\left\{w_i^n \in T_{\delta}(W_i|x_i^n)\right\}} \mathbbm{1}_{\left\{ \UCodeword = w_1^n\right\}}\mathbbm{1}_{\left\{ \VCodeword = w_1^n\right\}}\11_{\{w_1^n \oplus w_2^n,m_1,m_2\}} \nonumber 
\end{align}
where we define $\11_{\{w_1^n \oplus w_2^n,m_1,m_2\}} $ as 
\begin{align}
    \11_{\{w^n,m_1,m_2\}} \deq \11 \left\{ \exists (\tilde{w}^n, \tilde{a}): \tilde{w}^n G + h_1^{(\mu_1)}(m_1) + h_2^{(\mu_2)}(m_2), \tilde{w}^n \in \TDelta(W_1 \oplus W_2), \tilde{w}^n \neq w^n \right\}. \nonumber 
\end{align}
Using this we obtain, 
\begin{align}
    \EE[{S}_2] & \leq 2 \sum_{\underline{x}^n}
     \sum_{\substack{ \mu_1,\mu_2  }}\sum_{\substack{ m_1 > 0, \\m_2>0  }}\sum_{\substack{a_1 \in \FF_p^k\\a_2 \in \FF_p^k}} \sum_{\substack{w_1^n,w_2^n \in \FF_p^n}} \!\frac{p^{n}_{\underline{X}}(\underline{x}^n)}{2^{n(C_1+C_2)}} \frac{p^np^n}{2^{n(S_1+S_2)}(1+\eta)^2}{p_{W_1|X_1}^n(w_1^n|x_1^n)} {p_{W_2|X_2}^n(w_2^n|x_2^n)}  \nonumber \\ &\hspace{55pt} \mathbbm{1}_{\left\{w_i^n \in T_{\delta}(W_i|x_i^n)\right\}}\mathbbm{1}_{\left\{w_i^n \in T_{\delta}(W_i|x_i^n)\right\}} \mathbbm{1}_{\left\{ \UCodeword = w_1^n\right\}}\mathbbm{1}_{\left\{ \VCodeword = w_1^n\right\}}\11_{\{w_1^n \oplus w_2^n,m_1,m_2\}} \nonumber  
\end{align}
Note that, we have
\begin{align}
    \EE&\left[\11_{\{w_1^n \oplus w_2^n,m_1,m_2\}} \mathbbm{1}_{\left\{ \UCodeword = w_1^n\right\}}\mathbbm{1}_{\left\{ \VCodeword = w_1^n\right\}} \right] \nonumber \\
    & \leq \sum_{\tilde{a} \neq a}\sum_{\substack{\tilde{w}^n \in \TDelta(W_1 \oplus W_2) \\ \tilde{w}^n \neq w_1^n \oplus w^n_2}} \frac{1}{p^n}\frac{1}{p^n}\frac{1}{p^n} \leq 2^{n(H(W_1 \oplus W_2) + \delta_z)}p^{3n -k},
\end{align}
where $\delta_z(\delta) \searrow 0 $ as $\delta \searrow 0.$ This gives
\begin{align*}
    \EE[{S}_2] & \leq 2 \sum_{\underline{x}^n}
     \sum_{\substack{ m_1 > 0, \\m_2>0  }}\sum_{\substack{a_1 \in \FF_p^k\\a_2 \in \FF_p^k}} \sum_{\substack{w_1^n,w_2^n \in \FF_p^n}} {p^{n}_{\underline{X}}(\underline{x}^n)}\frac{2^{n(H(W_1 \oplus W_2) + \delta_z)}p^{n-k}}{2^{n(S_1+S_2)}(1+\eta)^2}{p_{W_1|X_1}^n(w_1^n|x_1^n)} {p_{W_2|X_2}^n(w_2^n|x_2^n)}  \nonumber \\ 
     & \leq  2^{n\left[(S_1 - R_1 ) - (\log{p} H(W_1 \oplus W_2) - \delta_{S_2})\right] },
\end{align*}
where $\delta_{S_2}(\delta) \searrow 0 $ as $\delta \searrow 0.$ 

\subsection{Proof of Proposition \ref{prop:Lemma for J1}} \label{proof_prop:Lemma for J1}
Substituting the definition of $\bar{E}^{(\mu_1)}_{W^n_1|X_1^{n}}(\cdot|\cdot)$ and $\gammaCoeff$ in $J_1$, we obtain
\begin{align*}
        J_1 & = \sum_{\underline{x}^n,w_2^n,y^n}\Bigg|p^n_{\underline{X}W_2Y}(\underline{x}^n,y^n) - \frac{1}{2^{nC_1}}\sum_{\substack{ \mu_1}}\sum_{\substack{w_1^n}} \sum_{m_1 > 0}\sum_{a_1 \in \FF_{q}^k} {p^{n}_{\underline{X}}(\underline{x}^n)}
        \frac{p^n}{2^{nS_1}(1+\eta)}{p_{W_1|X_1}^n(w_1^n|x_1^n)} \nonumber \\
     &\hspace{1.5in}\mathbbm{1}_{\left\{w_1^n \in T_{\delta}(W_1|x_1^n)\right\}}\11_{\{\UCodeword = w_1^n\}} p^n_{W_2|X_2}(w^n_2|x_2^{n})
p^n_ { Y|{Z}}(y^{n}|w_1^n+w_2^n)\Bigg| \\
& = \sum_{\underline{x}^n,w_2^n,y^n}\Bigg|p^n_{\underline{X}W_2Y}(\underline{x}^n,y^n) \nonumber \\
     & \hspace{0.5in} - \frac{1}{(1+\eta)}\frac{p^n}{2^{n(S_1+C_1)}}\sum_{\substack{ \mu_1}}\sum_{\substack{w_1^n \in T_{\delta}(W_1|x_1^n)}} \sum_{m_1 > 0}\sum_{a_1 \in \FF_{q}^k} {p^{n}_{\underline{X}\underline{W}Y}(\underline{x}^n , \underline{w}^n,y^n)} \11_{\{\UCodeword = w_1^n\}} \Bigg| \\
     & \leq J_{11} + J_{12}, 
\end{align*}
 where
 \begin{align*}
     J_{11} & \deq \sum_{\underline{x}^n,w_2^n,y^n} \Bigg|p^n_{\underline{X}W_2Y}(\underline{x}^n,y^n) - \frac{1}{(1+\eta)}\frac{p^n}{2^{n(S_1+C_1)}} \sum_{\substack{\mu_1}} \sum_{\substack{w_1^n}}   \sum_{\substack{m_1 > 0\\a_1 \in \FF_{q}^k}} {p^{n}_{\underline{X}\underline{W}Y}(\underline{x}^n, \underline{w}^n, y^n)} \11_{\{\UCodeword = w_1^n\}} \Bigg| \\
    J_{12} & \deq \sum_{\underline{x}^n,w_2^n,y^n}\Bigg|  \frac{1}{(1+\eta)}\frac{p^n}{2^{n(S_1+C_1)}}\sum_{\substack{ \mu_1}}\sum_{\substack{w_1^n \notin \TDelta(W_1)}} \sum_{m_1 > 0}\sum_{a_1 \in \FF_{q}^k} {p^{n}_{\underline{X}\underline{W}Y}(\underline{x}^n , \underline{w}^n,y^n)} \11_{\{\UCodeword = w_1^n\}} \Bigg|
 \end{align*}
 
 As for the term $J_{11}$, we use Lemma \ref{lem:changeMeasureCoveringLemma} and obtain the following bound on $\EE[J_{11}]$ as, $\EE[J_{11}] \leq \epsilon_{J_{11}}$ if $S_1 + C_1 \geq I(W_1;X_1X_2ZW_2) + \log{p} - H(W_1) + \delta_{J_{11}}.$
 
For the term $J_{12},$ applying expectation gives 
\begin{align*}
    \EE[J_{12}] \leq \frac{1}{(1+\eta)}\sum_{\underline{x}^n,w_2^n,y^n} \sum_{w_1^n \notin \TDelta(W_1)}{p^{n}_{\underline{X}\underline{W}Y}(\underline{x}^n , \underline{w}^n,y^n)} \leq \epsilon'. 
\end{align*}
where $\epsilon'(\delta) \searrow 0$ as $\delta \searrow 0.$

\subsection{Proof of Proposition \ref{prop:Lemma for Q2}}
\label{proof_prop:Lemma for Q_2}
We begin by 
applying triangle inequality on $Q_2$ to obtain $Q_2 \leq F_1 + F_2$, where
\begin{align*}
    F_1 & \deq \sum_{\underline{x}^n,y^n}\Bigg|\frac{1}{2^{n(C_1+C_2)}}\sum_{\substack{ \mu_1,\mu_2 }}\sum_{\substack{w_1^n,w_2^n \in \FF_p^n}} {p^{n}_{\underline{X}}(\underline{x}^n)}\gammaCoeff E^{(\mu_1)}_{W^n_1|X_1^{n}}(w^n_1|x_1^{n}) \nonumber \\ & \hspace{80pt}\left(p^n_{W_2|X_2}(w^n_2|x_2^{n})\! - \zetaCoeff \bar{E}^{(\mu_2)}_{W^n_2|X_2^{n}}(w^n_2|x_2^{n})\right)
p^n_ { Y|{Z}}(y^{n}|w_1^n\oplus w_2^n) \Bigg|, \\
F_2 & \deq \sum_{\underline{x}^n,y^n}\Bigg|\frac{1}{2^{n(C_1+C_2)}}\sum_{\substack{ \mu_1,\mu_2  }}\sum_{\substack{w_1^n,w_2^n \in \FF_p^n}} {p^{n}_{\underline{X}}(\underline{x}^n)}\gammaCoeff E^{(\mu_1)}_{W^n_1|X_1^{n}}(w^n_1|x_1^{n}) \nonumber \\ & \hspace{70pt}\zetaCoeff\left( \bar{E}^{(\mu_2)}_{W^n_2|X_2^{n}}(w^n_2|x_2^{n}) -  {E}^{(\mu_2)}_{W^n_2|X_2^{n}}(w^n_2|x_2^{n})\right)
p^n_ { Y|{Z}}(y^{n}|w_1^n\oplus w_2^n) \Bigg|,
\end{align*}
where $ \bar{E}^{(\mu_2)}_{W^n_2|X_2^{n}}(\cdot|\cdot)$ is defined as
\begin{align*}
   \bar{E}^{(\mu_2)}_{W^n_2|X_2^{n}}(w^n_2|x_2^{n}) & \deq  \frac{p^n}{2^{nS_2}(1+\eta)}{p_{W_2|X_2}^n(w_2^n|x_2^n)} \mathbbm{1}_{\left\{w_2^n \in T_{\delta}(W_2|x_2^n)\right\}}
\end{align*}
Considering the term corresponding to $F_1$, we bound it using triangle inequality applied by adding and subtracting the following terms within its modulus:
\begin{align*}
    (i) \;& \sum_{w_1^n,w_2^n \in \FF_p^n} {p^{n}_{\underline{X}}(\underline{x}^n)}p^n_{W_1|X_1}(w^n_1|x_1^{n}) p^n_{W_2|X_2}(w^n_2|x_2^{n})
p^n_ { Y|{Z}}(y^{n}|w_1^n\oplus w_2^n) \\
(ii) \; & \frac{1}{2^{nC_2}} \sum_{\mu_2} \sum_{w_1^n,w_2^n \in \FF_p^n} {p^{n}_{\underline{X}}(\underline{x}^n)}p^n_{W_1|X_1}(w^n_1|x_1^{n}) \frac{p^n\zetaCoeff}{2^{nS_2} (1+\eta)}  p^n_{W_2|X_2}(w^n_2|x_2^{n})
p^n_ { Y|{Z}}(y^{n}|w_1^n\oplus w_2^n) \\
(iii) \; &  \frac{1}{2^{n(C_1+C_2)}} \sum_{\mu_1, \mu_2} \sum_{w_1^n,w_2^n \in \FF_p^n}\!\!\!\!\!\! {p^{n}_{\underline{X}}(\underline{x}^n)}\gammaCoeff E^{(\mu_1)}_{W^n_1|X^n_1}(w^n_1|x_1^{n}) \frac{p^n\zetaCoeff}{2^{nS_2} (1+\eta)}\  p^n_{W_2|X_2}(w^n_2|x_2^{n})
p^n_ { Y|{Z}}(y^{n}|w_1^n\oplus w_2^n) 
\end{align*}
This gives the following bound $F_1 \leq F_{11} + F_{12} + F_{13} + F_{14},$ where
\begin{align*}
    F_{11} & \deq \sum_{\underline{x}^n,y^n}\Bigg|\frac{1}{2^{nC_1}}\sum_{\substack{ \mu_1 }}\sum_{\substack{w_1^n,w_2^n \in \FF_p^n}} {p^{n}_{\underline{X}}(\underline{x}^n)}\left(\gammaCoeff E^{(\mu_1)}_{W^n_1|X_1^{n}}(w^n_1|x_1^{n}) - p^n_{W^n_1|X_1^{n}}(w^n_1|x_1^{n}) \right)p^n_ { Y|{Z}}(y^{n}|w_1^n\oplus w_2^n) \Bigg| \\ 
    F_{12} & \deq \sum_{\underline{x}^n,y^n}\Bigg| \sum_{w_1^n,w_2^n \in \FF_p^n} {p^{n}_{\underline{X}}(\underline{x}^n)}p^n_{W_1|X_1}(w^n_1|x_1^{n})\left( p^n_{W_2|X_2}(w^n_2|x_2^{n})\right. \\ & \hspace{2.7in} \left. - \frac{p^n\zetaCoeff}{2^{nS_2} (1+\eta)}  p^n_{W_2|X_2}(w^n_2|x_2^{n}) \right)
p^n_ { Y|{Z}}(y^{n}|w_1^n\oplus w_2^n)\Bigg| \\
        F_{13} & \deq \sum_{\underline{x}^n,y^n}\Bigg| \frac{1}{2^{nC_2}} \sum_{\mu_2} \sum_{w_1^n,w_2^n \in \FF_p^n} {p^{n}_{\underline{X}}(\underline{x}^n)}\left(p^n_{W_1|X_1}(w^n_1|x_1^{n}) \right. \\ & \hspace{1in} \left. - \frac{1}{2^{nC_1}} \sum_{\mu_1} \gammaCoeff E^{(\mu_1)}_{W^n_1|X^n_1}(w^n_1|x_1^{n})\right) \frac{p^n\zetaCoeff}{2^{nS_2} (1+\eta)}  p^n_{W_2|X_2}(w^n_2|x_2^{n})
    p^n_ { Y|{Z}}(y^{n}|w_1^n\oplus w_2^n)  \Bigg| \\
    F_{14} & \deq \sum_{\underline{x}^n,y^n}\Bigg|\frac{1}{2^{n(C_1+C_2)}} \sum_{\mu_1, \mu_2} \sum_{\substack{w_1^n\in \FF_p^n}}\sum_{w_2^n  \notin \TDelta(W_2)}{p^{n}_{\underline{X}}(\underline{x}^n)}\gammaCoeff \bar{E}^{(\mu_1)}_{W^n_1|X^n_1}(w^n_1|x_1^{n}) \frac{p^n\zetaCoeff}{2^{nS_2} (1+\eta)}\\ & \hspace{3.5in}   p^n_{W_2|X_2}(w^n_2|x_2^{n}) 
p^n_ { Y|{Z}}(y^{n}|w_1^n\oplus w_2^n)  \Bigg| \\
\end{align*}
We start by analyzing $F_{11}.$ Note that $F_{11}$ is exactly similar to the term $Q_1$ and hence using the same rate constraints as $Q_1$, this term can be bounded. Next consider the term corresponding to $F_{12}.$ Substituting the definition of $\zetaCoeff$ gives
\begin{align*}
    F_{12} & =  \sum_{\underline{x}^n,y^n}\Bigg|  {p^{n}_{\underline{X}Y}(\underline{x}^n,y^n)} -  \frac{p^n}{2^{nS_2} (1+\eta)} \sum_{w_2^n \in \FF_p^n} \sum_{\substack{m_2 > 0\\a_2 \in \FF_{q}^k}} \11_{\{\VCodeword = w_2^n\}}  p^{n}_{\underline{X}W_2Y}(\underline{x}^n,w_2^n,y^n) \Bigg|
\end{align*}
Lemma \ref{lem:changeMeasureCoveringLemma} gives us functions $\epsilon_{F_{12}}(\delta), \delta_{F_{12}}(\delta) $ such that if $$S_1 \geq I(W_2;X_1X_2Y) - H(W_2) + \log{p} + \delta_{F_{12}},$$ then $\EE[F_{12}] \leq   \epsilon_{F_{12}}$, where $\epsilon_{F_{12}}(\delta), \delta_{F_{12}}(\delta) \searrow 0 $ as $\delta \searrow 0.$

\noindent Now, we move on to considering in the term corresponding to $F_{13}.$ Taking expectation with respect to $G, h_1^{(\mu_1)}$ and $h_2^{(\mu_2)}$ gives
\begin{align*}
\EE[F_{13}] & = \EE_{G,h_1}\left[ \sum_{\underline{x}^n,y^n}\Bigg| \frac{1}{2^{nC_2}} \sum_{\mu_2} \sum_{w_1^n,w_2^n \in \FF_p^n} {p^{n}_{\underline{X}}(\underline{x}^n)}\bigg(p^n_{W_1|X_1}(w^n_1|x_1^{n}) \right. 
    \\ & \hspace{0.3in} \left. - \frac{1}{2^{nC_1}} \sum_{\mu_1} \gammaCoeff E^{(\mu_1)}_{W^n_1|X^n_1}(w^n_1|x_1^{n})\bigg) \EE_{h_2|G}\left[\frac{p^n\zetaCoeff}{2^{nS_2} (1+\eta)} \right] p^n_{W_2|X_2}(w^n_2|x_2^{n})
p^n_ { Y|{Z}}(y^{n}|w_1^n\oplus w_2^n)\right]\\
& = \frac{1}{(1+\eta)}\EE_{G,h_1}\left[ \sum_{\underline{x}^n,y^n}\Bigg| \frac{1}{2^{nC_2}} \sum_{\mu_2} \sum_{w_1^n,w_2^n \in \FF_p^n} {p^{n}_{\underline{X}}(\underline{x}^n)}\bigg(p^n_{W_1|X_1}(w^n_1|x_1^{n}) \right. 
    \\ & \hspace{1.5in} \left. - \frac{1}{2^{nC_1}} \sum_{\mu_1} \gammaCoeff E^{(\mu_1)}_{W^n_1|X^n_1}(w^n_1|x_1^{n})\bigg) p^n_{W_2|X_2}(w^n_2|x_2^{n})
p^n_ { Y|{Z}}(y^{n}|w_1^n\oplus w_2^n)\right]\\
& = \EE\left[\frac{J}{(1+\eta)}\right],
\end{align*}
where the above equalities follows from the fact that $h_1^{(\mu_1)}$ and $h_1^{(\mu_1)}$ were  generated independently and from using the definition of $J$ as stated earlier. Therefore, using the same analysis and rate constraints as $J$, we can bound the term $F_{13}$. Finally, we remain with the term $F_{14}$. Applying expectation on $F_{14}$ gives
\begin{align*}
    \EE\left[F_{14}\right] & \leq \EE_{G,h_2}\left[\sum_{\underline{x}^n}\Bigg| \sum_{\mu_1, \mu_2} \!\!\!\!\!\sum_{\substack{w_1^n\in \FF_p^n\\w_2^n  \in \TDelta(W_2)}}\!\!\!\! {p^{n}_{\underline{X}}(\underline{x}^n)}\frac{E_{h_1}\left[\gammaCoeff\right]}{2^{n(C_1+C_2)}} E^{(\mu_1)}_{W^n_1|X^n_1}(w^n_1|x_1^{n}) \frac{p^n\zetaCoeff}{2^{nS_2} (1+\eta)} p^n_{W_2|X_2}(w^n_2|x_2^{n})\Bigg| \right]\\
    & \leq \sum_{\underline{x}^n} \frac{1}{2^{nC_2}} \sum_{\mu_2} \sum_{\substack{w_1^n\in \TDelta(W_1)\\w_2^n  \notin \TDelta(W_2)}}\!\!\!\! {p^{n}_{\underline{X}}(\underline{x}^n)} \frac{p^n_{W_1|X_1}(w^n_1|x_1^{n})}{(1+\eta)} \EE\left[\frac{p^n\zetaCoeff}{2^{nS_2} (1+\eta)}\right] p^n_{W_2|X_2}(w^n_2|x_2^{n}) \\
   & \leq \frac{1}{(1+\eta)^2}\sum_{\underline{x}^n}  \sum_{\substack{w_2^n  \notin \TDelta(W_2)}}\!\!\!\! {p^{n}_{\underline{X}}(\underline{x}^n)}   p^n_{W_2|X_2}(w^n_2|x_2^{n}) \leq \epsilon_w',
\end{align*}
where $\epsilon_w'(\delta) \searrow 0$ as $\delta \searrow 0.$ This completes the analysis for the term corresponding to $F_1$. Finally we remain with the analysis of the term $F_{2}.$ Simplifying $F_2$ gives
\begin{align}
    F_2 & \leq \frac{1}{2^{n(C_1+C_2)}} \sum_{\substack{ \mu_1,\mu_2  }} \sum_{\underline{x}^n} {p^{n}_{\underline{X}}(\underline{x}^n)}
    \Bigg(\sum_{\substack{w_1^n\in \FF_p^n}} \gammaCoeff E^{(\mu_1)}_{W^n_1|X_1^{n}}(w^n_1|x_1^{n})\Bigg)   \nonumber \\ & \hspace{70pt}\Bigg|\sum_{\substack{w_2^n \in \FF_p^n}}\zetaCoeff\left( \bar{E}^{(\mu_2)}_{W^n_2|X_2^{n}}(w^n_2|x_2^{n}) -  {E}^{(\mu_2)}_{W^n_2|X_2^{n}}(w^n_2|x_2^{n})\right) \Bigg|, \nonumber \\
& = \frac{1}{2^{nC_2}} \sum_{\substack{ \mu_1,\mu_2  }} \sum_{\underline{x}^n} {p^{n}_{\underline{X}}(\underline{x}^n)}
    \bigg(\sum_{m_1 > 0}p_{M_1|X^n_1}(m_1|x_1^n)  \bigg) \nonumber \\ & \hspace{70pt}\Bigg|\sum_{\substack{w_2^n \in \FF_p^n}}\zetaCoeff\left( \bar{E}^{(\mu_2)}_{W^n_2|X_2^{n}}(w^n_2|x_2^{n}) -  {E}^{(\mu_2)}_{W^n_2|X_2^{n}}(w^n_2|x_2^{n})\right) \Bigg|, \nonumber \\
  & \leq \frac{1}{2^{nC_2}} \sum_{\substack{ \mu_1,\mu_2  }} \sum_{\underline{x}^n} {p^{n}_{\underline{X}}(\underline{x}^n)}
  \Bigg|\sum_{\substack{w_2^n \in \FF_p^n}}\zetaCoeff\left( \bar{E}^{(\mu_2)}_{W^n_2|X_2^{n}}(w^n_2|x_2^{n}) -  {E}^{(\mu_2)}_{W^n_2|X_2^{n}}(w^n_2|x_2^{n})\right) \Bigg|, \nonumber \\
    = \tilde{S}_{2}, \nonumber
\end{align}
where the last inequality above follows by using $\bigg(\sum_{m_1 > 0}p_{M_1|X^n_1}(m_1|x_1^n)  \bigg) \leq 1$ and the last equality follows by recalling the definition of $\tilde{S}_2$. Therefore, using the constraints obtained in the analysis of $\tilde{S}_2$, we complete the proof of the proposition.

\newpage
\bibliographystyle{IEEEtran}
\bibliography{references}

\begin{thebibliography}{10}
\providecommand{\url}[1]{#1}
\csname url@samestyle\endcsname
\providecommand{\newblock}{\relax}
\providecommand{\bibinfo}[2]{#2}
\providecommand{\BIBentrySTDinterwordspacing}{\spaceskip=0pt\relax}
\providecommand{\BIBentryALTinterwordstretchfactor}{4}
\providecommand{\BIBentryALTinterwordspacing}{\spaceskip=\fontdimen2\font plus
\BIBentryALTinterwordstretchfactor\fontdimen3\font minus
  \fontdimen4\font\relax}
\providecommand{\BIBforeignlanguage}[2]{{%
\expandafter\ifx\csname l@#1\endcsname\relax
\typeout{** WARNING: IEEEtran.bst: No hyphenation pattern has been}%
\typeout{** loaded for the language `#1'. Using the pattern for}%
\typeout{** the default language instead.}%
\else
\language=\csname l@#1\endcsname
\fi
#2}}
\providecommand{\BIBdecl}{\relax}
\BIBdecl

\bibitem{atif2020ISITsource}
T.~{Anwar Atif}, A.~{Padakandla}, and S.~{Sandeep Pradhan}, ``Source coding for
  synthesizing correlated randomness,'' in \emph{2020 IEEE International
  Symposium on Information Theory (ISIT)}, 2020, pp. 1570--1575.

\bibitem{korner1979encode}
J.~Korner and K.~Marton, ``How to encode the modulo-two sum of binary sources
  (corresp.),'' \emph{IEEE Transactions on Information Theory}, vol.~25, no.~2,
  pp. 219--221, 1979.

\bibitem{ahlswede1983source}
R.~Ahlswede and T.~Han, ``On source coding with side information via a
  multiple-access channel and related problems in multi-user information
  theory,'' \emph{IEEE Transactions on Information Theory}, vol.~29, no.~3, pp.
  396--412, 1983.

\bibitem{padakandla2016achievable}
A.~Padakandla, A.~G. Sahebi, and S.~S. Pradhan, ``An achievable rate region for
  the three-user interference channel based on coset codes,'' \emph{IEEE
  Transactions on Information Theory}, vol.~62, no.~3, pp. 1250--1279, 2016.

\bibitem{197401TIT_Wyn}
A.~Wyner, ``Recent results in the shannon theory,'' \emph{Information Theory,
  IEEE Transactions on}, vol.~20, no.~1, pp. 2 -- 10, Jan 1974.

\bibitem{CuffPhDThesis}
P.~W. Cuff, ``Communication in networks for coordinating behavior,'' Ph.D.
  dissertation, Stanford, CA, USA, 2009.

\bibitem{201311TIT_Cuf}
P.~{Cuff}, ``Distributed channel synthesis,'' \emph{IEEE Transactions on
  Information Theory}, vol.~59, no.~11, pp. 7071--7096, Nov 2013.

\bibitem{yassaee2015channel}
M.~H. Yassaee, A.~Gohari, and M.~R. Aref, ``Channel simulation via interactive
  communications,'' \emph{IEEE Transactions on Information Theory}, vol.~61,
  no.~6, pp. 2964--2982, 2015.

\bibitem{winter2004extrinsic}
A.~Winter, ``‘‘{E}xtrinsic’’and ‘‘{I}ntrinsic’’ data in quantum
  measurements: Asymptotic convex decomposition of positive operator valued
  measures,'' \emph{Communications in mathematical physics}, vol. 244, no.~1,
  pp. 157--185, 2004.

\bibitem{wilde2012information}
M.~M. Wilde, P.~Hayden, F.~Buscemi, and M.-H. Hsieh, ``The
  information-theoretic costs of simulating quantum measurements,''
  \emph{Journal of Physics A: Mathematical and Theoretical}, vol.~45, no.~45,
  p. 453001, 2012.

\bibitem{201907ISIT_HeiAtiPra}
M.~{Heidari}, T.~A. {Atif}, and S.~{Sandeep Pradhan}, ``Faithful simulation of
  distributed quantum measurements with applications in distributed
  rate-distortion theory,'' in \emph{2019 IEEE International Symposium on
  Information Theory (ISIT)}, July 2019, pp. 1162--1166.

\bibitem{krithivasan2011distributed}
D.~Krithivasan and S.~S. Pradhan, ``Distributed source coding using abelian
  group codes: A new achievable rate-distortion region,'' \emph{IEEE
  Transactions on Information Theory}, vol.~57, no.~3, pp. 1495--1519, 2011.

\bibitem{200710TIT_NazGas}
B.~Nazer and M.~Gastpar, ``Computation over multiple-access channels,''
  \emph{IEEE Trans. on Info. Th.}, vol.~53, no.~10, pp. 3498 --3516, oct. 2007.

\bibitem{200906TIT_PhiZam}
T.~Philosof and R.~Zamir, ``On the loss of single-letter characterization: The
  dirty multiple access channel,'' \emph{IEEE Trans. on Info. Th.}, vol.~55,
  pp. 2442--2454, June 2009.

\bibitem{201109TITarXiv_JafVis}
A.~Jafarian and S.~Vishwanath, ``Achievable rates for k-user {G}aussian
  interference channels,'' submitted to {I}{E}{E}{E} {T}rans. of {I}nformation
  theory 2011, available at {\tt http://arxiv.org/abs/1109.5336}.

\bibitem{pradhanalgebraic}
S.~S. Pradhan, A.~Padakandla, and F.~Shirani, \emph{An Algebraic and
  Probabilistic Framework for Network Information Theory}.\hskip 1em plus 0.5em
  minus 0.4em\relax Foundations and Trends in Communications and Information
  Theory, 2021, vol.~18, no.~2.

\bibitem{padakandla2013computing}
A.~Padakandla and S.~S. Pradhan, ``Computing sum of sources over an arbitrary
  multiple access channel,'' in \emph{2013 IEEE International Symposium on
  Information Theory}.\hskip 1em plus 0.5em minus 0.4em\relax IEEE, 2013, pp.
  2144--2148.

\bibitem{Gal-ITRC68}
R.~G. Gallager, \emph{Information Theory and Reliable Communication}.\hskip 1em
  plus 0.5em minus 0.4em\relax New York: John Wiley \& Sons, 1968.

\bibitem{atif2020source}
T.~A. Atif, A.~Padakandla, and S.~S. Pradhan, ``Source coding for synthesizing
  correlated randomness,'' \emph{arXiv preprint arXiv:2004.03651}, 2020.

\bibitem{cuff2009communication}
P.~W. Cuff, \emph{Communication in networks for coordinating behavior}.\hskip
  1em plus 0.5em minus 0.4em\relax Stanford University, 2009.

\bibitem{fourier-motzkin}
G.~M. Ziegler, \emph{Lectures on polytopes}.\hskip 1em plus 0.5em minus
  0.4em\relax Springer Science \& Business Media, 2012, vol. 152.

\end{thebibliography}

\end{document}